%% file: efficient_separation.tex
\documentclass[11pt,letterpaper]{article}

\usepackage[margin=1in]{geometry}
\usepackage{setspace}
\usepackage{amsmath}
\usepackage{mathtools}
\usepackage{tikz}
\usepackage{amsfonts}
\usepackage{amsthm}
\usepackage{graphicx}
\usepackage{hyperref}
\usepackage[linesnumbered,ruled,vlined]{algorithm2e}
\usepackage{booktabs}
\usepackage{bm}
\usepackage{authblk}
\usepackage[numbers,sort&compress]{natbib}
\usepackage{amssymb}
\usepackage[T1]{fontenc}
\usepackage{xcolor}

\usepackage{lmodern}
\usepackage{microtype}
\usepackage{enumitem}

\newcommand{\mindisagree}{\textsc{MinDisagree}}

\newcommand{\cost}{\operatorname{cost}}
\newcommand{\poly}{\operatorname{poly}}
\newcommand{\one}{\mathbf 1}
\newcommand{\OPT}{\operatorname{OPT}}

\newcommand{\CCMinRatio}{\textsc{CC-MinRatio}}

\newcommand{\mycomment}[1]{}

\newcommand{\ratioconst}{1.3865}

\newtheorem{theorem}{Theorem}[section]
\newtheorem{lemma}[theorem]{Lemma}
\newtheorem{corollary}[theorem]{Corollary}
\newtheorem{proposition}[theorem]{Proposition}
\newtheorem{fact}{Fact}

\theoremstyle{definition}
\newtheorem{definition}[theorem]{Definition}
\theoremstyle{plain}

\theoremstyle{remark}
\newtheorem{remark}[theorem]{Remark}

\theoremstyle{plain}
\newtheorem{problem}[theorem]{Problem}

\title{Approximate Dual Separation for the Cluster LP: a $1.387$-approximation for Correlation Clustering}

\author[1]{David Garc\'ia-Soriano\thanks{david.garcia.soriano@upc.edu}}
\author[2]{Antoine Schohn\thanks{antoine.schohn@polytechnique.edu}}

\affil[1]{Universitat Polit\`ecnica de Catalunya, Barcelona, Spain}
\affil[2]{\'Ecole Polytechnique, Palaiseau, France}

\date{\today}

\begin{document}
\maketitle

\begin{abstract}
\input{abstract}
\end{abstract}
\thispagestyle{empty}

\clearpage
\tableofcontents
\thispagestyle{empty}

\clearpage
\pagenumbering{arabic}

\section{Introduction}
Correlation Clustering is perhaps the most natural formulation of clustering, given general pairwise similarity information between objects.
Pairwise similarity judgments create conflicting clustering preferences: an object may be deemed similar to each of two other objects which are
themselves dissimilar to each other.
The problem seeks a globally consistent partition that best reconciles these preferences, without prescribing the number of clusters in advance.
Its input is a complete signed graph, with positive (resp., negative) edges representing similar (resp., dissimilar) pairs. Its output is a partition of the
vertices minimizing the number of disagreements (the \mindisagree\ objective function): negative pairs placed in
the same cluster or positive pairs placed in different clusters.

The computational complexity of the problem and its applications to machine learning have been thoroughly investigated; see~\cite{Bonchi+22,cc_tutorial} and the references therein.
In its contemporary form, it was introduced by~\citet*{Bansal+04}, although early formulations date at least as far back
as the 1960s~\cite{zahn1964approximating,regnier1965quelques}. The problem is NP-hard~\cite{kvrivanek1986np,wakabayashi1986aggregation,Bansal+04}, motivating the search
 for approximation algorithms. The first polynomial-time constant-factor guarantee was given by~\citet*{Bansal+04}; \citet*{Charikar+05} later obtained a $4$-approximation. \citet*{Ailon+08}
 achieved ratio $3$ via a simple combinatorial algorithm,
 and gave an LP-based refinement with ratio $2.5$. Subsequently, \citet*{Chawla+15} achieved a $2.06$-approximation, near the factor-$2$ integrality gap of the metric LP
 relaxation of~\citep{Charikar+05}.
Thus, any rounding guarantee below $2$ cannot be
measured against the metric-LP value.

Recently,
a breakthrough of
\citet*{Cohen-Addad+22} broke this apparent barrier, giving a
$(1.994+\varepsilon)$-approximation
via the Sherali-Adams hierarchy. This spurred a sequence of improvements:
\citet*{Cohen-Addad+23} reduced the ratio to
$1.73+\varepsilon$ and introduced the concept of preclustering, rooted in the technique of~\cite{cohen2021correlation}.
Preclustering identifies certain vertex pairs that can be assumed to be together (or separated) in a nearly-optimal solution, and  bounds the number of remaining undecided pairs in
terms of the optimum cost.
Later \citet*{Cao+24}
introduced
the \emph{cluster LP}.
Whereas the metric LP represents fractional separation through pairwise distances, the cluster~LP assigns fractional weights to entire candidate clusters, ensuring each vertex is covered by total weight one.
At the price of exponentially many variables, this provides information about joint cluster membership that is absent from pairwise distances alone.
Remarkably, the latter work succeeded in
approximating the cluster LP,
and
derived the state-of-the-art
$(1.485+\varepsilon)$-approximation algorithm for \mindisagree\ (see~\cite{Cao+24arXiv} for a clarification about the approximation ratio achieved).
Subsequent work gave, for fixed $\varepsilon$, $\widetilde O(n)$-time algorithms: a combinatorial $(1.847+\varepsilon)$-approximation~\cite{Cohen-Addad+24}, and a cluster-LP
implementation achieving $(1.485+\varepsilon)$~\cite{Cao+25}; these algorithms are sublinear in the size of the signed graph, $\Theta(n^2)$.
The proven guarantees of both cluster-LP solvers~\cite{Cao+24,Cao+25} compare the value of the computed primal solution with the integral optimum.

\subsection{Our contributions}\label{sec:contrib}
Our results address two complementary questions: how accurately the fractional cluster LP can be solved, and how effectively its solutions can be rounded into clusterings.
We give an approximation scheme for the former and a $\ratioconst$-factor rounding algorithm for the latter.

The cluster LP is an exponential-size relaxation with one variable for every possible cluster, akin to the \emph{configuration LPs} used for bin
packing~\cite{bin_packing}, the Santa Claus problem~\cite{Bansal+06} and generalized assignment~\cite{Fleischer+06}, which have one variable per feasible bin or bundle.
Such LPs are usually solved through their duals, which have only polynomially many variables: the ellipsoid method reduces the task to separation.
For those problems, exact separation is a knapsack-type problem, which is NP-hard but admits a fully polynomial-time approximation scheme; this suffices to solve the LP to any
degree of accuracy.
For the cluster LP, separation is of a different nature: given signed vertex weights~$q$, one must find a set $S$ of vertices whose total weight $q(S)$ exceeds $\cost(S)$, the disagreements attributed to $S$ in any clustering having $S$ as a cluster.
This is a quadratic optimization problem over vertex sets, neither submodular nor supermodular.
Like knapsack, it is NP-hard to solve exactly (Theorem~\ref{thm:hardness-intro}); unlike knapsack, it remains hard on instances with polynomially bounded integral weights, and therefore admits no
fully polynomial-time approximation scheme unless $\mathrm{P}=\mathrm{NP}$.

We give the first approximation scheme for the fractional cluster LP, computing approximately optimal primal and dual solutions.
We achieve this through approximate dual separation:
our oracle
either finds a violated dual constraint or reports feasibility after
scaling the dual vector by
$1+\varepsilon$.
This
makes the cluster LP algorithmically accessible as a relaxation in its own right, so
we can approximate the \emph{fractional} cluster-LP optimum, produce a feasible dual solution certifying this approximation, recover a sparse feasible primal solution, and (more
        generally) optimize every linear objective over the dual polytope approximately.
This differs from the proven guarantees of Cao et al.~\cite{Cao+24,Cao+25}, whose algorithms first precluster the instance and construct a feasible {primal} solution of value
at most $(1+\varepsilon)$ times the \emph{integral} optimum. They remark that their primal analysis should extend to the fractional optimum, but do not provide this extension;
separately, they state that they do not know how to solve the cluster LP by considering its dual~\cite[Section~1.1 and footnote~1]{Cao+24arXiv}. Our result supplies an explicit self-contained dual-based proof needing no preclustering.
\begin{theorem}[Approximate dual separation]\label{thm:dual-intro}
For every $\varepsilon>0$, there is a deterministic approximate separation algorithm for the cluster-LP dual that either returns a violated constraint or
correctly reports that the given dual vector becomes feasible after division by $1+\varepsilon$. Its running time is $2^{\poly(1/\varepsilon)}(n+\langle q\rangle)^{O(1)}$, where $\langle q\rangle$ denotes the bit length of $q$. Thus, in time $2^{\poly(1/\varepsilon)}n^{O(1)}$ one can approximate the fractional cluster-LP optimum within a factor $1+\varepsilon$ and produce a feasible dual solution together with a feasible primal solution of polynomial support certifying this approximation.
\end{theorem}
A key component of the proof is our \emph{localization} technique, a purely graph-theoretic tool of independent interest (Lemma~\ref{lem:threshold-localization}): given the size of an unknown vertex set $S$ such that every bipartition $S=A\mathbin{\dot\cup}B$ has at least $\delta|A||B|$ crossing edges (for a constant $\delta>0$), and most of whose vertices have degree $O(|S|)$, it constructs at most $n$ candidate universes, one of which contains $S$ and has size $O(|S|)$.
This allows us to apply dense optimization techniques~\cite{delavega96,arora_dense99,frieze_kannan_quick}
 to a small
candidate universe. Whereas a direct application of these tools on the entire graph gives
an additive error proportional to $n^2$, within a localized universe the
additive error scales as $|S|^2$ instead.

Our next result
shows that the use of approximate separation is necessary unless $\mathrm{P}=\mathrm{NP}$, thus clarifying the picture of the approximability of the cluster LP.
   To the best of our
   knowledge, no hardness result for solving the fractional cluster LP has appeared in the literature.
\begin{theorem}[Hardness of exact separation and optimization]\label{thm:hardness-intro}
Given a complete signed graph $G=(V,E^+)$ and a threshold $T$,
deciding whether its fractional cluster-LP optimum is at most $T$ is NP-complete.
Furthermore, deciding exact dual feasibility of a vector $q\in\mathbb Q^V$ is coNP-complete.
\end{theorem}
The situation is drastically different for incomplete graphs.
Since the integrality gap of the cluster LP is at most 2 (as the analysis of cluster-based rounding carries over to incomplete graphs),
a polynomial-time constant-factor approximation to it is ruled out under the unique games conjecture by~\cite{hardness_multicut,Demaine+06}.
On the other hand, an $O(\log n)$-factor approximation follows from the rounding algorithm of~\cite{Charikar+05,Demaine+06} for the metric LP, the value of which is no larger
    than that of
    the cluster LP.

Our last main result is a rounding algorithm for every feasible cluster-LP solution.

\begin{theorem}[Rounding the cluster LP]\label{thm:rounding-intro}
There is a polynomial-time deterministic rounding algorithm for the cluster LP with approximation factor $\ratioconst$.
\end{theorem}
As a consequence, the integrality gap of the cluster LP is at most $\ratioconst$; combined with the known lower bound of $4/3$~\cite{Cao+24}, this places the gap in the narrow
interval $[4/3,\ratioconst]$.
Thus, our rounding closes about two thirds of the distance between the previous best ratio and the known integrality-gap lower bound.
For reference, the best known hardness-of-approximation factor for \mindisagree{} under randomized reductions is $24/23\approx1.043$~\cite{Cao+24}.

\begin{corollary}
For every $\varepsilon>0$, there is a deterministic $(\ratioconst+\varepsilon)$-factor approximation algorithm for correlation clustering in time $2^{\poly(1/\varepsilon)}\poly(n)$.
\end{corollary}

Classical pivot algorithms decide whether a vertex joins the cluster of the pivot from the sign and the LP value of the pair alone~\cite{Ailon+08,Chawla+15}.
The cluster LP additionally allows sampling a cluster $S$ containing the pivot, and in the pivot rounding of~\cite{Cao+24arXiv} each positive pair either follows $S$ or ignores it,
according to its LP value.
Our rule interpolates between these two modes: the joining probability is an affine function of the indicator of membership in $S$, with coefficients that vary continuously with the LP value.
In the analysis, the factor-revealing SDP is replaced by an explicit sum of variances that implies the guarantee  after verifying finitely many polynomial inequalities.
We also note that we can replace our rounding procedure by an alternative method based solely on dual separation (Remark~\ref{rem:rounding-separation}).

Our rounding can be combined with either cluster-LP solver, with different benefits. Together with
    our first result,
    it gives an instance-dependent certificate: the ratio between the cost of the clustering returned
    and the value of the feasible dual solution is $\le \ratioconst+\varepsilon$.
Alternatively, the primal solution produced by Cao et al.~\cite{Cao+25}
can be fed to our rounding algorithm, giving a $(\ratioconst+\varepsilon)$-approximation against the integral optimum in time $\widetilde O(2^{\poly(1/\varepsilon)}n)$ with adjacency-list queries.
This faster composition cannot provide the per-instance primal--dual certificate:
approximate separation requires $\Omega(n^2)$ adjacency-list queries, even for non-negative vectors $q$ (see Appendix~\ref{sec:query_lb}).\footnote{This does not contradict the sublinear covering-MWU primitive discussed by Cao et al.~\cite{Cao+25},
designed for a shifted dual and
whose guarantee is relative to an integral clustering. Nor does it rule out sublinear algorithms that output a primal solution of value at most $(1+\varepsilon)\OPT_{\mathrm{LP}}$; such algorithms cannot
output a matching feasible dual solution.}
In particular, no approximate separation algorithm with adjacency-list access runs in sublinear time.

Note that our polynomial-time $(\ratioconst+\varepsilon)$-approximation extends to graphs with edge weights satisfying the probability constraints ($w^+_{uv}, w^-_{uv}\ge 0$ and
$w^+_{uv}+w^-_{uv} = 1$): the black box reduction in the full version of~\cite[Theorem~20]{Chawla+15} reduces any such instance to an instance of the unweighted problem. The
algorithm becomes randomized in this case.

\paragraph{Additional applications.} As our first application, we formulate a natural semi-supervised clustering method with seed labels. In \emph{seeded correlation clustering},
the input contains the signed graph and an arbitrary set $R$ of seeds.
Each seed acts as the representative of a distinct cluster, but we do not assume that all ground-truth clusters have an associated seed.
We require that every cluster output contain at most one seed; our objective is
still to minimize the number of disagreements.
We compute a $(1+\varepsilon)$-approximate primal solution to the analogous cluster LP in $2^{\poly(1/\varepsilon)}n^{O(1)}$ time.  Combined with the rounding theorem
of~\citet{towards2} for constrained correlation clustering, this gives a randomized $(1.92+\varepsilon)$-approximation for the seeded problem.

Our second application implements the cluster-insertion search from the
combinatorial local-search framework of~\citet{Cohen-Addad+24}, via a direct reduction to approximate separation.
Our reduction gives a deterministic polynomial-time approximate insertion oracle requiring no preclustering;
the existing implementation in~\cite{Cohen-Addad+24} needed preclustering, and sampling enabled their faster running time.
Although
the guarantee for correlation clustering obtained through this route is $\approx 1.847$,
    this illustrates the power of localization to optimize an exponentially large search.

\paragraph{Organization.}
Section~\ref{sec:dual} defines the cluster LP and its dual and states the separation theorem (Theorem~\ref{thm:main}); Section~\ref{sec:overview} gives a technical
overview of our results.
Sections~\ref{sec:extreme-regimes}--\ref{sec:together} prove Theorem~\ref{thm:main}, and Section~\ref{sec:rounding} presents and analyzes our rounding scheme.
The appendices contain the deferred proofs (Appendices~\ref{app:extreme-regime-proofs}, \ref{sec:dense} and~\ref{app:rounding-verification}), the hardness results
(Appendix~\ref{app:hardness}), our two applications (Appendix~\ref{app:applications}), the rounding derandomization (Appendix~\ref{app:rounding-derandomization}), and a covering-gap example (Appendix~\ref{app:covering-gap}).

\subsection{Related work}
   The weighted version of Correlation Clustering, in which pairs carry arbitrary
   nonnegative weights, is equivalent to minimum multicut and is $O(\log n)$-factor approximable, as observed by~\citet*{Charikar+05} and~\citet*{Demaine+06}.  Assuming the Unique Games Conjecture, no polynomial-time $O(1)$-factor approximation algorithm exists in this case~\cite{hardness_multicut}.

~\citet{active_queries} gave a query-efficient algorithm where pairwise similarities can be consulted; their running time is exponential in $n$.
A sublinear-time approximation algorithm with a mixed multiplicative-additive guarantee was designed by~\citet{local_corr} and simplified in~\cite{Garcia-Soriano+20}; see also~\cite{Kuroki+24,cold_start}.
The work of~\cite{Cao+25,Cohen-Addad+24} gave sublinear-time algorithms with multiplicative guarantees, assuming the stronger adjacency-list oracle.

Many  other variants have been investigated, including fixed numbers of clusters~\cite{Giotis+06_j}, maximizing
agreements~\cite{Charikar+05,Swamy04,hermite_laptop}, constrained cluster sizes~\cite{Puleo+15},
    overlapping clusters~\cite{Bonchi+13}, edge-labeled graphs~\cite{Bonchi+15}, multilayer networks~\cite{Miyauchi+26}, hypergraphs~\cite{Kim+11,Li+17},
    parallel and streaming settings~\cite{Ahn+15,Chierichetti+14,behnezhad_streaming23,cao_rounds24}, sketches~\cite{cc_desparsification}
and hierarchical correlation clustering~\cite{rounding_hierarchical}.

    Finally, we note that equivalent problems have been researched under names such as \emph{cluster editing}~\citep{bocker13cluster} and \emph{clique
        partitioning}~\cite{cutting_plane,grotschel1990facets}.

\section{The cluster LP and its separation problem}\label{sec:dual}
                                                                  We encode
the positive pairs by an unweighted graph $G=(V,E^+)$ and write
$E^-=\binom{V}{2}\setminus E^+$ for the complementary set of negative pairs.
For $\sigma\in\{+,-\}$, let $e^\sigma(S)$ denote the number of pairs of
$E^\sigma$ with both endpoints in $S$, and, for disjoint $A,B\subseteq V$,
let $e^\sigma(A,B)$ denote the number
of pairs of
$E^\sigma$
with one endpoint in each set.
All degrees refer to the positive-edge graph $G=(V,E^+)$:
$d_A(u)=e^+(\{u\},A\setminus \{u\})$ denotes the number of positive neighbors of $u$ in $A$, and
$d(u)=d_V(u)$ is its positive degree.
For $u\in V$, its \emph{closed positive neighborhood} is
$\Gamma^+(u)=\{u\}\cup\{v\in V:uv\in E^+\}$.
Define

\begin{equation}\label{eq:pseudo-boolean}
  \cost(S)=e^-(S)+\frac12e^+(S,V\setminus S)
  =\binom{|S|}{2}+\frac12\sum_{u\in S}d(u)-2e^+(S).
\end{equation}
This is precisely the contribution of a cluster $S$ to the correlation clustering objective (when minimizing disagreements):
an internal negative pair is counted once, while the cost of an external positive pair
is split between the two clusters it crosses.
    With this definition, the total number of disagreements of a partition $\mathcal P = \{S_1, \ldots, S_k\}$ of $V$ can be decomposed as the sum $\sum_{S \in {\mathcal P}} \cost(S)$.
The cluster LP and its dual are:
{\small
\begin{center}
\fbox{\begin{tabular}{@{}p{0.46\textwidth}@{\hspace{0.6em}\vrule\hspace{0.6em}}p{0.46\textwidth}@{}}
    \vspace{-0.5cm}
\begin{minipage}[t]{\linewidth}
\begin{equation}\label{eq:cluster-primal}
\begin{aligned}
  \text{minimize}\quad
  &\sum_{S\subseteq V}\cost(S)z_S,\\
  \text{subject to}\quad
  &\sum_{S\ni u}z_S=1 &&\forall u\in V,\\
  &z_S\ge0 &&\forall S\subseteq V.
\end{aligned}
\end{equation}
\end{minipage}
&
    \vspace{-0.5cm}
\begin{minipage}[t]{\linewidth}
\begin{equation}\label{eq:cluster-dual}
\begin{aligned}
  \text{maximize}\quad &\sum_{u\in V}q_u,\\
  \text{subject to}\quad &\sum_{u \in S} q_u \le\cost(S) &&\forall S\subseteq V,\\
  &q_u\in \mathbb{R}&&\forall u\in V.
\end{aligned}
\end{equation}
\end{minipage}
\end{tabular}}
\end{center}
}
These are not a covering/packing pair of problems; the coordinates of $q$ are unrestricted because the corresponding constraints in \eqref{eq:cluster-primal} are equalities.
This distinction is important: Appendix~\ref{app:covering-gap} gives an eight-vertex instance
where every optimal dual solution has a negative coordinate.

Our main result is an approximate separation oracle,
    in
the form needed to run the ellipsoid algorithm for linear  programming~\cite{ellipsoid}.

\begin{theorem}[Approximate separation]\label{thm:main}
There is a deterministic algorithm which, given an
$n$-vertex graph $G=(V,E^+)$ and $q\in\mathbb Q^V$, $\varepsilon > 0$,
either returns a set $T\subseteq V$ with
    $
  q(T)>\cost(T),
  $
or reports that $q/(1+\varepsilon)$ is feasible for the dual
\eqref{eq:cluster-dual}.
It
runs in time
$
  2^{\poly(1/\varepsilon)}(n+\langle q\rangle)^{O(1)}.
  $
\end{theorem}
Notice that, were $\cost()$ submodular, minimization of $\cost(S) - q(S)$ would yield an exact dual separation oracle. However, this does not hold due to the
supermodular term $e^-(S)$ in~\eqref{eq:pseudo-boolean} (the remaining term being a submodular cut function).

It is convenient to consider the underlying cost ratio problem.  For a
signed vector $q\in\mathbb Q^V$, put
$q^+(A)=\sum_{u\in A} \max\{q_u,0\}.$
Assume that $q^+(V)>0$ (i.e., some vertex has positive $q$-weight).

\begin{problem}[\CCMinRatio]
\label{prob:cc-minratio}
Given $G=(V,E^+)$ and a signed vector $q\in\mathbb Q^V$ with $q^+(V)>0$, find a
set $S\subseteq V$ with $q(S)>0$ minimizing
    $
  \rho(S)=\frac{\cost(S)}{q(S)}.
  $
\end{problem}

We denote the optimum value of \CCMinRatio{} by
$
  \rho^*=\min_{S:\,q(S)>0}\rho(S).
  $
  The vector $q$ is feasible for the dual cluster LP if and only if $\rho^* \ge 1$.

\begin{corollary}[EPTAS for \CCMinRatio]\label{cor:cc-minratio-eptas}
For every $\varepsilon>0$, there is an algorithm which, given an
$n$-vertex graph $G=(V,E^+)$ and $q\in\mathbb Q^V$ with $q^+(V)>0$, returns a set
$T\subseteq V$ with $q(T)>0$ and
$
  \rho(T)\le(1+\varepsilon)\rho^*
  $
in time
$2^{\poly(1/\varepsilon)}(n+\langle q\rangle)^{O(1)}$.
\end{corollary}
Corollary~\ref{cor:cc-minratio-eptas} follows by applying Theorem~\ref{thm:main} to the rescaled vector $\lambda q$ and performing binary search on $\lambda$; the
zero-optimum case
(a clique component of positive $q$-weight)
    is recognized directly.

\begin{corollary}[Approximate primal and dual solutions]\label{cor:lp-solutions}
For every $\varepsilon>0$, there is an algorithm running in time
$2^{\poly(1/\varepsilon)}n^{O(1)}$ that
computes both a feasible vector $q$ for~\eqref{eq:cluster-dual} and a
feasible solution $z$ of~\eqref{eq:cluster-primal} of polynomial support such
that
\[
  \sum_{u\in V}q_u\ge\frac{\OPT_{\mathrm{LP}}}{1+\varepsilon}
  \qquad\text{and}\qquad
  \sum_S\cost(S)z_S\le(1+\varepsilon)\OPT_{\mathrm{LP}}.
\]
Here $\OPT_{\mathrm{LP}}$ is the common optimum value of
\eqref{eq:cluster-primal} and~\eqref{eq:cluster-dual}.
\end{corollary}
Corollary~\ref{cor:lp-solutions} follows from Theorem~\ref{thm:main} by the ellipsoid method~\cite{ellipsoid} and a standard argument for approximate separation
oracles~\cite{bin_packing,Fleischer+06}, with a sufficiently smaller internal accuracy $\varepsilon'$ to make room for optimization precision.
Dividing an accepted dual vector by the oracle's approximation factor makes it exactly feasible.
To recover the primal solution, let $\mathcal S$ contain all constraints returned during the ellipsoid execution together with all singletons, and consider the dual relaxation
containing only the constraints in $\mathcal S$.  The same ellipsoid execution is valid for this relaxation: every separating constraint returned belongs to $\mathcal S$, and since
every accepted vector becomes feasible for the full dual after division by $1+\varepsilon'$, the oracle never accepts a vector of
objective value above $(1+\varepsilon')\OPT_{\mathrm{LP}}$.
The correctness of the ellipsoid algorithm
therefore certifies, using only the constraints in $\mathcal S$, that the optimum of the
relaxation is at most $(1+\varepsilon)\OPT_{\mathrm{LP}}$; the singleton constraints, together with the objective cut $\sum_uq_u\ge0$, bound the search region.  Solving the corresponding restricted primal gives a polynomial-support feasible solution of
the same value.
The case $\OPT_{\mathrm{LP}}=0$, in which the positive graph is a disjoint union of cliques, is recognized directly.

\section{Technical overview}\label{sec:overview}
\subsection{Approximate dual separation}\label{sec:separation-overview}
We sketch here the main ideas of the proof of Theorem~\ref{thm:main}.
Fix an optimal solution $S^*$ of \CCMinRatio, and let
$
  \rho^*=\rho(S^*)=\frac{\cost(S^*)}{q(S^*)}
  $
denote its objective value.
If $\rho^*\ge1/(1+\varepsilon)$, then $q/(1+\varepsilon)$ is dual feasible.
Thus, the oracle needs only find a violated constraint
in the gap case
\begin{equation}\label{eq:gap-case}
  \rho^*<\frac{1}{1+\varepsilon}.
\end{equation}
In this case, $q(S^*)>(1+\varepsilon)\cost(S^*)$, and hence the Lagrangian
$
  L(T)=\cost(T)-q(T)
    $
satisfies
$
  L(S^*)<-\varepsilon \cost(S^*).
  $
To find a violated dual constraint,
it suffices to find an approximate minimizer of $L$ with error below the margin cost $\varepsilon \cost(S^*)$, yielding a set with negative $L$-value.

Our starting point is
to build
a constant-sized ``summary'' of the graph
approximating the number of edges between two arbitrary subsets
using the weak regularity lemma of
Frieze and
Kannan~\cite{frieze_kannan_quick}.
We further refine the classes so that the linear part of $L$ (which depends on $q$) is nearly constant within each class.
The quadratic minimization problem for $L$ can be solved within the summary, and its solution can then be lifted to
the original graph. However, the additive error
incurred
is
proportional to $|V|^2$,
and can therefore
exceed the $\varepsilon\cost(S^*)$ margin. Thus, we introduce a {localization} technique whose
purpose is precisely to reduce the universe to $O_\varepsilon(|S^*|)$ vertices without losing the negative witness.
This idea works well when $\cost(S^*)/|S^*|^2$ is neither too small nor too high.

Let
$
  s=|S^*|,\ c=\cost(S^*).
    $
The
following three regimes apply:
\begin{center}
\begin{tabular}{@{}lll@{}}
\toprule
cost of $S^*$ & useful structure & algorithmic step \\
\midrule
$c \gg s^2/\varepsilon$
& a singleton is nearly optimal & test positive-weight singletons \\
$c \ll \varepsilon s^2$
& $S^*$
is close to one neighborhood & neighborhood + linearization \\
$ c \in [\Omega(\varepsilon s^2), O(s^2/\varepsilon)]$
& additive error is affordable & localization + weak regularity \\
\bottomrule
\end{tabular}
\end{center}

\paragraph{Parameters.}
We set the following internal parameters for our algorithm:
\begin{equation}\label{eq:parameters}
  \beta=\frac{\varepsilon}{10},\qquad
  \tau=\frac{\varepsilon}{12800}.
\end{equation}

\paragraph{Large cost.}
Localization is inefficient in this case, so we handle it separately.
The cost of $S^*$ is dominated by the positive edges outgoing from $S^*$, so $c$ roughly equals half the sum of all degrees in $S^*$.
Moreover, the negative part of $q$ is negligible, otherwise deleting it can only increase the cost by $O(s^2)$, but the overall ratio would be improved. Quantitatively, by
Lemma~\ref{lem:structure}\ref{it:positive-mass},
$
  \rho^*q^+(S^*)
  \le\binom{s}{2}+c.
  $
When $c\ge s(s-1)/(2\beta)$, some positive-weight singleton $u \in S^*$ has ratio at most
\[
    \rho(\{u\}) \le
  \frac{\binom{s}{2}+c}{q^+(S^*)}
  \le\rho^*\left(1+\frac{s(s-1)}{2c}\right)
  \le(1+\beta)\rho^*.
\]
As $\beta < \varepsilon$,
testing every
positive-weight singleton finds a violated dual constraint  in the gap case.

\paragraph{Small cost.}
    Here even the error from the localized universe would overwhelm the margin, so we take a different route.
The identity (Lemma~\ref{lem:structure}\ref{it:pivots})
\[
  \sum_{v\in S^*}|\Gamma^+(v)\mathbin\triangle S^*|=2\cost(S^*)
\]
shows that, when $c\le\tau s^2$, some positive neighborhood is
$O(\varepsilon s)$-close to $S^*$.  A degree threshold around that neighborhood produces
a superset $U\supseteq S^*$ with $|U\setminus S^*|=O(c/s)$.  Among its $s$-subsets,
a linear
approximation to $L$ can be minimized greedily.  The discarded quadratic correction costs only
$O(|U\setminus S^*|^2)=O(\tau c)$, below the $\varepsilon c$ gap
margin.

\paragraph{Intermediate cost.}
Perhaps the most interesting case is the remaining one.
Optimality of $S^*$ implies that in any partition $S^*=A\mathbin{\dot\cup}B$, there will be many positive edges between $A$ and $B$. In particular, the following key
half-density bound holds (Lemma~\ref{lem:structure}\ref{it:half-density}):
\[
  e^+(A,B)\ge\frac{|A||B|}{2}.
\]
Our {localization} procedure constructs,
   without knowing $S^*$, a small
vertex set $U$ containing $S^*$.
The set $U$ itself need not have negative $L$-value, but our search for a negative witness can be restricted to a subset of $U$.
Keeping $U$ small is essential to control an additive error proportional to $|U|^2$.
Thus, in this regime
   localization plays a role analogous to that of preclustering in earlier algorithms: it makes additive approximation errors small relative to the relevant
solution cost.

In a certain sense, localization is a set-cover problem with hidden ground set $S^*$, whose covering sets are the positive neighborhoods of its vertices.  To motivate the construction, imagine an ideal adaptive process that starts from a pivot in $S^*$ and repeatedly chooses another pivot from the part of $S^*$ already exposed.  Half density makes the first neighborhood cover at least half of $S^*$ and suggests that further well-chosen neighborhoods should quickly cover the remainder.  This process cannot be implemented directly, however: the algorithm does not know which exposed vertices belong to $S^*$, while treating every exposed vertex as a new center can allow incidental neighbors outside $S^*$ to make the visible universe grow uncontrollably.

We replace the choice of a new center by a deterministic propagation rule.  Starting from the positive neighborhood of each possible low-degree pivot, we repeatedly expose every
vertex having at least $s/20$ positive neighbors among the exposed low-degree vertices.  This is loosely analogous to replacing a greedy set-cover choice by threshold rounding.  Half density guarantees progress despite not knowing $S^*$: the first threshold closure captures a constant fraction of the still-hidden vertices of $S^*$, after which every remaining vertex of $S^*$ crosses the threshold and a second closure captures it.  At the same time, every newly exposed vertex accounts for at least $s/20$ edge incidences from low-degree centers, so double counting bounds the size increase in each round.  Trying every possible initial pivot therefore produces a family containing some $U\supseteq S^*$ with $|U|=O_\varepsilon(s)$.

It remains to recover a negative-$L$ subset from $U$.  For $X\subseteq U$, let $x\in\{0,1\}^U$ be its incidence vector.  Expanding the definition of $\cost(X)$ gives
\begin{align}\label{eq:Lx}
  L(X)=\sum_{\{u,v\}\in\binom{U}{2}}(1-2\cdot\one[uv\in E^+])x_ux_v
       +\sum_{u\in U}\left(\frac{d(u)}2-q_u\right)x_u.
\end{align}
Thus the restriction of $L$ to $U$ is a quadratic pseudo-Boolean function.  Write $b_u=d(u)/2-q_u$.  If some $b_u<0$, the singleton
$\{u\}$ itself is a violation.  Otherwise $b\ge0$, and the weak regularity lemma of
Frieze and Kannan can be used to approximate the number of edges between any two subsets.
A nearly-optimal minimizer of $L$ on $U$ is then determined, up to a small error, by how many vertices it takes from each weakly-regular class (the linear term decides
which ones: those with the smallest $b_u$), and can be found by solving
the problem
on a ``summary'' graph whose size depends only on $\varepsilon$.
This approximation
has
an additive error $O(\varepsilon |U|^2)$, smaller than the negative margin of $L(S^*)$, thereby returning a set
of negative $L$-value.

Trying all $s$ (guesses for $|S^*|$) and pivots covers the three possibilities.  The only
exceptional case is $\rho^*=0$, which is easily recognizable by finding a
positive-weight connected component forming a clique.
Sections~\ref{sec:extreme-regimes} to~\ref{sec:together} formalize the sketch above, completing the
proof of Theorem~\ref{thm:main}.

\subsection{Rounding: conditional-affine pivots and a variance certificate}\label{sec:rounding-overview}
Similarly to the rounding algorithm of Cao et al.~\cite{Cao+24arXiv}, our method combines cluster-based rounding (which is particularly effective on negative edges)
    with a pivot-based rounding procedure
(the details of which differ).
In a pivot round, after choosing a pivot $u$, we sample an LP cluster $S\ni u$. Write $y_{uv}=\Pr[v\in S]$ for the LP co-membership value of a pair $uv$ of sign $\sigma$; see
~\eqref{eq:marginals}. Our new ingredient is a conditional-affine joining rule under which $v$ joins the pivot with conditional probability
\[
  \mu_{\sigma(uv)}(y_{uv})+\gamma_{\sigma(uv)}(y_{uv})\bigl(\mathbf 1[v\in S]-y_{uv}\bigr)
\]
for judiciously chosen functions $\mu_\sigma, \gamma_\sigma$ that depend on the edge sign $\sigma(uv)$ and its LP value $y_{uv}$.

For positive pairs, our specific instantiation of $\mu$ and $\gamma$, detailed in ~\eqref{eq:analytic-rules},
always adds members of $S$ and independently ``rescues'' outsiders with a probability that grows
linearly from $0$ to $1$ as $y_{uv}$ ranges from $1/2$
up to a certain value; negative pairs are rounded independently of $S$. Thus the positive rule continuously interpolates between following the sampled cluster and ignoring it,
   strictly generalizing the pure dependent/independent modes of~\cite[Algorithm~2]{Cao+24arXiv}.
   The final algorithm combines the cluster-based rounding of~\cite{Cao+24arXiv} with our conditional-affine pivot rounding from
   Algorithm~\ref{alg:conditional-pivot}, and returns the cheaper output; its cost
   is analyzed
   via a convex combination.
\begin{algorithm}[H]
\caption{Conditional-affine pivot rounding $\mathcal P$}\label{alg:conditional-pivot}
\KwIn{Signed graph $G$; a feasible cluster-LP solution $z$, with marginals $y$ given by~\eqref{eq:marginals}.}
$W\gets V$; $\mathcal A\gets\emptyset$\;
\While{$W\ne\emptyset$}{
  choose a pivot $u\in W$ uniformly at random\;
  sample $S\ni u$ with probability $z_S$\;
  $C\gets\{u\}$\;
  \ForEach{$v\in W\setminus\{u\}$}{
    $A_v\gets\mathbf 1[v\in S],\quad \sigma\gets\sigma(uv)$\;
    independently add $v$ to $C$ with probability $\mu_\sigma(y_{uv})+\gamma_\sigma(y_{uv})(A_v-y_{uv})$\;
  }
  $\mathcal A\gets\mathcal A\cup\{C\}$; $W\gets W\setminus C$\;
}
\Return{$\mathcal A$}\;
\end{algorithm}

The starting point of the analysis is the per-round accounting framework of Cao et al.~\cite[Section~2.2]{Cao+24arXiv}, based on edge budgets that are released when an endpoint is removed.
Their framework reduces the desired approximation guarantee to showing that an aggregate one-round surplus (i.e., the allowance prescribed by the target ratio minus the expected disagreement cost) is nonnegative.
Concretely, a pair $e$ has LP cost $\ell_e$ ($1-y_e$ if positive, $y_e$ if negative), and its budget $B_e$ is derived
from the
target approximation ratio
and the expected cost ($\kappa_e$) of cluster-based rounding on $e$.
The surplus decomposes into local contributions indexed by signed triangles and pairs.
The contribution of some feasible triangles is negative,
so the task is to show that these deficits cannot occur in an arbitrary combination; this is
where our analysis departs from that of~\cite{Cao+24arXiv}.

Conditioned on a fixed pivot $u$, let $A_v=\mathbf 1[v\in S]$ and write $y_{uvw}=\mathbb E[A_vA_w]$. Then $\mathbb E[A_v]=y_{uv}$ and $y_{uvw}-y_{uv}y_{uw}$ is the covariance of
$A_v$ and $A_w$. Writing $W$ for the vertices alive, we choose explicit closed-form weights $\Theta_\sigma(y)$ and consider
\[
  Z_u:=\sum_{v\in W\setminus\{u\}}\Theta_{\sigma(uv)}(y_{uv})(A_v-y_{uv}).
\]
The cross terms in $\sum_{u\in W}\mathbb E[Z_u^2]$ are indexed by triangles, and the diagonal terms are indexed by pairs. We prove that every local surplus contribution dominates
the corresponding term in this expansion. Summing shows that the aggregate surplus is at least $\sum_{u\in W}\mathbb E[Z_u^2]\ge0$. The weights $\Theta$ are given by explicit
formulas.
The remaining local polynomial inequalities are verified over every signed triangle and pair type using exact arithmetic.
In contrast to the approach
of Cao et al., this argument uses no global triangle-frequency variables, cellwise surplus bounds, or a factor-revealing SDP: it gives a closed-form sum-of-variances certificate for our more flexible pivot rule.

\paragraph{Design of the rounding rules.}
The closed forms for $\mu, \gamma$ and $\Theta$ in~\eqref{eq:analytic-rules} and~\eqref{eq:theta-closed-form} were obtained numerically and then simplified.  Having fixed the analysis above, we
discretized the $y$ values into $20$ intervals and gave each (interval, sign) class its own conditional-affine rule and
certificate weight, so that a computer could test a candidate table against the local inequalities.  A search by local improvements produced tables certifying
ratios down to $1.392$.  The optimized tables were close to simple curves;
fitting closed forms to these patterns by polynomial regression on two intervals
gave~\eqref{eq:analytic-rules} and~\eqref{eq:theta-closed-form}.  The closed forms improve on the tables they came from ($\ratioconst$ versus $1.392$), because rules,
budgets and weights become exact functions of $y$ rather than per-cell approximations.  Once fixed, it only remained to prove
their correctness
(Lemma~\ref{lem:verification}), irrespective of their origin.

\begin{remark}[Rounding by dual separation]\label{rem:rounding-separation}
Interestingly, there is an alternative rounding method based on dual separation, using neither cluster-based nor pivot-based rounding. It gives a (slower) deterministic
$(\ratioconst+\varepsilon)$-rounding in polynomial time for fixed $\varepsilon$, and does not need the cluster LP variables $z_S$:
it only needs to
know the projected edge
variables $y_{uv}$ and the budget functions $B_e$ and $\kappa_e$.
The key is to use dual separation to construct the clusters whose existence is guaranteed by the analysis above.
For an alive set $W$ and a nonempty $C\subseteq W$, let $d_W(C)=e^-(C)+e^+(C,W\setminus C)$ be the cost incurred when $C$ is output as a cluster.
By~\eqref{eq:one-round-accounting}, some $C$ releases at least as much budget as it costs: $d_W(C)\le\sum_{e\cap C\ne\emptyset}B_e$ (where sums are restricted to pairs inside $W$);
the same holds with $\kappa_e$ in place of $B_e$.
    Finding such a cluster is another dual separation problem
    to which the bounded-weight extension in Appendix~\ref{app:cluster-insertion} applies.\footnote{More precisely, enlarge the budgets by $\varepsilon\ell_e$ and express the released budget minus $d_W(C)$ as a quadratic function of the membership indicators. Replace pairwise coefficients of absolute value below a sufficiently small $\delta=\Theta(\varepsilon)$ by $-\delta$. Every modified pair has $\ell_e=\Omega(1)$, so the lost surplus is $O(\delta)\sum_{e\cap C\ne\emptyset}\ell_e$ and is absorbed by the added slack.
        The auxiliary weights then lie in $[\delta,O(1)]$; after rescaling, the bounded-weight separator applies with accuracy $\Theta(\varepsilon)$.}
Repeatedly removing such a cluster, for whichever budget has the smaller total on $W$, gives a clustering of cost at most $\min\{\sum_e\kappa_e,\sum_eB_e\}+\varepsilon\sum_e\ell_e\le(\ratioconst+\varepsilon)\sum_e\ell_e$.
This suggests a route to improving the approximation: suppose that every instance could be shown to contain a cluster with $d_W(C)\le\gamma\sum_{e\cap C\ne\emptyset}\ell_e$ for some
fixed $\gamma$, even if the existence proof provides no efficient sampling procedure.  Dual separation then makes such existence arguments constructive, and greedily removing
clusters of approximately minimum ratio $d_W(C)/\sum_{e\cap C\ne\emptyset}\ell_e$  would round the cluster LP within a factor of $\gamma+\varepsilon$ (without knowing $\gamma$).
Whether this rule can improve the rounding factor further remains open.
\end{remark}

\subsection{NP-hardness through balanced completion}\label{sec:hardness-overview}
We sketch  here the hardness proof for the cluster LP; the same construction works for hardness of dual separation.

It is convenient to establish hardness for
incomplete graphs first, using a reduction from \textsc{Vertex Cover}.
Given a graph
$J=(U,E)$, we add a root vertex $r$, declare every pair $ru$ positive and
every edge of $J$ negative, and leave all other pairs absent.  A clustering
may be assumed to consist of a root cluster $\{r\}\cup I$ and singleton
clusters outside it, yielding cost
$
|U|-|I|+e_J(I).
$
The minimum of this expression is exactly the vertex-cover number
$\tau(J)$;
the fractional cluster-LP
optimum can be shown to be also $\tau(J)$.

We transfer the hardness result to complete unweighted instances by removing the absent pairs via a balanced completion gadget. We replace every vertex by a positive
clique.  Specified positive and negative pairs become uniformly
signed complete bipartite graphs between the corresponding blocks, whereas
an absent pair is replaced by a balanced pattern having as many positive
as negative pairs in every row and column, using Hadamard matrices.  On unions of entire blocks, these
patterns contribute only an explicit quantity.
Their small spectral norm ensures that selecting
a block partially
cannot overcome the positive-clique penalty inside the
block.  Consequently, the completion preserves
the cluster-LP optimum up to a known additive term.

\subsection{Applications}\label{sec:seeded-overview}
\paragraph{Seeded CC.}
The seeded problem introduced in Section~\ref{sec:contrib} is the
special case of constrained correlation clustering~\cite{vanZuylen+09} in which the hostile pairs are $\binom R2$ and there are no friendly pairs.  The constrained problem
has a $3$-approximation~\cite{vanZuylen+09}, recently improved to $\approx 2.29$~\cite{azizeddin2026constrained}.  \citet{towards2} showed that a polynomial-support
approximately optimal solution to its cluster LP would yield a $(1.92+\varepsilon)$-approximation.
Recent work has shown that
    the general constrained problem is UGC-hard to approximate below $2$~\cite{azizeddin2026constrained,CaoXu26}.
Hence, carrying out this program in full generality for constrained correlation clustering would disprove the UGC (or $P\neq NP$).
    Our result supplies the missing LP approximation for the seeded subclass, yielding a better-than-$2$ guarantee in this restricted setting.

The separation proof extends cleanly because the admissible clusters form a downwards-closed family.
Hence, if $S^*$ is an optimum admissible ratio set, both sides of every partition $S^*=A\mathbin{\dot\cup}B$ remain admissible, and the optimality argument giving
half density continues to hold.  The localization construction itself is unchanged. Although the universe $U$ returned may contain several seeds, only the negative-$L$ subset
ultimately recovered from $U$ must be admissible.

Only the recovery step from each candidate universe must enforce the seed constraint.  In the small-cost regime, we compare the best fixed-cardinality set containing no seed with the best one containing a single seed.  In the intermediate regime, we refine every weak-regularity class according to seed membership and add one aggregate packing constraint requiring the recovered set to contain at most one seed.  The large-cost argument uses only singletons and needs no change.
The ellipsoid method (with the admissibility constraints) then gives an approximately optimal seeded cluster-LP solution of polynomial support. Applying the rounding theorem of~\citet{towards2} yields the $(1.92+\varepsilon)$-approximation stated above.
(Our own rounding does not apply directly here, as its pivot clusters need not satisfy the seeded constraint.)
Appendix~\ref{app:seeded-cc} gives the formal details.

\paragraph{Cluster insertion.}
Our second application is to implement the search for a cluster insertion improving the correlation clustering cost, as in the combinatorial local-search framework
of~\citet{Cohen-Addad+24} for standard correlation clustering.  Given a current clustering $\mathcal P$, cluster insertion extracts a set $S$ from its current clusters and makes $S$ a new cluster.  Because the change in cost is a sum of pairwise contributions, it can be regrouped as
\[
  \cost_w(\mathcal P+S)-\cost_w(\mathcal P)=\widehat\cost(S)-q(S),
\]
where $\widehat\cost$ is the cost of an auxiliary weighted instance and $q$ distributes the cost of the disagreements of $\mathcal P$ among their endpoints.  Thus, $S$ is an improving insertion exactly when $q(S)>\widehat\cost(S)$, that is, when $S$ violates a cluster-dual constraint of the auxiliary instance.  The auxiliary pair weights remain within a fixed constant range in the iterated-flipping framework, so a bounded-weight extension of our separator gives an approximate insertion oracle.
Our separator gives a deterministic, preclustering-free, polynomial-time implementation of this primitive.

\section{The large and small cost regimes}\label{sec:extreme-regimes}
In this section we handle the two extreme cost regimes, which are comparatively  simpler. The small-cost case is related to the intuition underlying the
preclustering procedure of Cohen-Addad, Lee, Li, and Newman~\cite{Cohen-Addad+23}: a cluster whose cost is small relative to its squared size is nearly a positive clique, so positive neighborhoods of its vertices approximate the cluster. We use this observation directly for the unknown optimum, without constructing a global preclustering. The weighted singleton argument and the reconstruction by linearization are specific to our ratio problem.
The proofs of all statements in this section are deferred to Appendix~\ref{app:extreme-regime-proofs}.

We collect below some key identities and consequences of the optimality of $S^*$.
\begin{lemma}\label{lem:structure}
The following statements hold.
\begin{enumerate}[label=(\alph*)]
  \item\label{it:global-support}
  For every $T\subseteq V$,
  \begin{equation}\label{eq:global-support}
    \cost(T)\ge\rho^* q(T).
  \end{equation}
  \item  For every pair of disjoint sets $A, B\subseteq V$,
    \begin{equation}\label{eq:partition-identity}
      \cost(A)+\cost(B)=\cost(A \cup B)+2e^+(A,B)-|A||B|.
    \end{equation}

  \item\label{it:half-density}
  For every partition $S^*=A\mathbin{\dot\cup}B$,
  \[
    e^+(A,B)\ge\frac{|A||B|}{2}.
  \]
  In particular, $d_{S^*}(u)\ge(|S^*|-1)/2$ for every $u\in S^*$ (take $A=\{u\}, B = S^*\setminus \{u\}$).

  \item\label{it:deletion}
  For every $T\subseteq V$ and $X\subseteq T$,
  \[
    \cost(X)\le\cost(T)+|X|\,|T\setminus X|.
  \]

  \item\label{it:pivots}
  For every $T\subseteq V$,
  \begin{equation}\label{eq:pivot-sum}
    \sum_{v\in T}|\Gamma^+(v)\mathbin\triangle T|=2e^-(T)+e^+(T,V\setminus T)=
    2\cost(T),
  \end{equation}
  \begin{equation}\label{eq:dvs}
    \sum_{v\in T}|d(v)-(|T|-1)|\le2\cost(T),\qquad e^+(T,V\setminus T)\le2\cost(T).
   \end{equation}

  \item\label{it:positive-mass}
  For every $T\subseteq V$, its positive $q$-weight satisfies
  \begin{equation}\label{eq:positive-mass}
    \rho^* q^+(T)
    \le \frac12\sum_{\substack{u\in T\\q_u>0}}d(u)
    \le \binom{|T|}{2}+\cost(T).
  \end{equation}
  Consequently, whenever $\cost(S^*)\ge\tau|S^*|^2$,
  \begin{equation}
    \frac{q^+(S^*)}{q(S^*)}\le1+\frac{|S^*|(|S^*|-1)}{2\cost(S^*)}\le1+\frac{1}{2\tau}.
  \end{equation}
\end{enumerate}
\end{lemma}
Item~\ref{it:half-density} is the key consequence of optimality, and its proof is short: by~\eqref{eq:global-support} and~\eqref{eq:partition-identity},
$\cost(S^*)+2e^+(A,B)-|A||B|=\cost(A)+\cost(B)\ge\rho^*\bigl(q(A)+q(B)\bigr)=\cost(S^*)$.

\subsection{Large-cost: singleton sets}\label{sec:structure}

\begin{corollary}\label{cor:singleton}
Some $u\in S^*$ with $q_u>0$ satisfies
\begin{equation}
  \rho(\{u\})\le\rho^*\left(1+\frac{|S^*|(|S^*|-1)}{2\cost(S^*)}\right).
\end{equation}
Consequently, if $\cost(S^*)\ge|S^*|(|S^*|-1)/(2\varepsilon)$, then some singleton is a $(1+\varepsilon)$-approximation.
\end{corollary}

\subsection{Small cost: near-clique reconstruction by linearization}\label{sec:near}

We use the fixed objective $L(T)=\cost(T)-q(T)$.  The candidate construction
comes from a simple identity for sets of a fixed cardinality.

\begin{lemma}\label{lem:fixed-cardinality}
Fix a target cardinality $t$, a set $U\subseteq V$, and a set
$D\subseteq U$ of size $|U|-t$.
Then
\begin{equation}\label{eq:fixed-cardinality}
  L(U\setminus D)=K_{U,t}+
  \sum_{u\in D}\left(2d_U(u)-\frac12d(u)+q_u\right)-2e^+(D),
\end{equation}
where the quantity $K_{U,t}$
is independent of $D$.
\end{lemma}

Thus, among fixed-cardinality deletions, the linear effect of deleting $u$ is
\begin{equation}
  \sigma_u=2d_U(u)-\frac12d(u)+q_u;
\end{equation}
the only nonlinear correction comes from the term $-2e^+(D)$.

For each guessed cardinality $t\in[n]$, each $v\in V$, let $A=\Gamma^+(v)$ and
define the
neighborhood
\begin{equation}\label{eq:near-U}
  U=\{u\in V:d_A(u)\ge t/3\}.
\end{equation}
If $|U|\ge t$,
   our procedure
   outputs the set $T$ of the $t$ vertices of $U$ with the highest score $\sigma_u$.

The intuition is as follows. The set $A$ itself may miss vertices of the unknown optimum, but since $A$ is close to $S^*$, every vertex of $S^*$ has many positive neighbors in $A$.
Thus, $S^*$ is entirely contained in the neighborhood $U$. Conversely, every vertex in $U$ from outside $S^*$ has many positive edges into $S^*$, which allows us to upper-bound the
number of such extra vertices in the low-cost regime.
This, in turn, bounds the error we make by discarding the $2 e^+(D)$ term during minimization of~\eqref{eq:fixed-cardinality}.

\begin{lemma}\label{lem:near}
Suppose
$
  0<\cost(S^*)\le\frac{|S^*|^2}{200}.
  $
Then, for some pivot $v$, the candidate constructed with $t=|S^*|$ satisfies
\[
  L(T)\le L(S^*)+\frac{64\cost(S^*)^2}{|S^*|^2}.
\]
\end{lemma}

\section{The intermediate regime}\label{sec:intermediate}
We now assume the remaining regime applies:
\begin{equation}\label{eq:intermediate-regime}
  \tau |S^*|^2< \cost(S^*)<\frac{1}{\beta}\binom{|S^*|}{2}.
\end{equation}

\input{localization_threshold}

\section{Approximate separation: putting it all together}\label{sec:together}

\begin{proof}[Proof of Theorem~\ref{thm:main}]
We may assume that $0<\varepsilon\le1/4$ and $q^+(V)>0$.
The algorithm first tests whether a connected component $K$ is a clique
and has $q(K)>0$.  If so, it returns $K$, which satisfies
$q(K)>\cost(K)=0$.  Otherwise it generates the following candidates:
\begin{enumerate}[label=(\roman*)]
  \item every singleton $\{u\}$ with $q_u>0$;
  \item for every $s\in[n]$ and $v\in V$, the near-clique candidate of
        Section~\ref{sec:near}, whenever it is defined;
  \item for every $s\in[n]$ and every $v\in L_s$, the output of the recovery procedure from Section~\ref{sec:lagrangian} applied to the threshold-localized universe $U(v)$ in~\eqref{eq:threshold-universe}.
\end{enumerate}
It evaluates $L(T)=\cost(T)-q(T)$ exactly for every candidate, returning one
with $L(T)<0$, if found.  If none has negative value, it returns
$q/(1+\varepsilon)$.

To prove correctness, let $S^*$ be an optimum, with
$s=|S^*|$, $c=\cost(S^*)$, and objective value $\rho^*$.  The initial test
handles the case $\rho^*=0$, so suppose now that
$\rho^*>0$.  If $\rho^*\ge1/(1+\varepsilon)$, then
Lemma~\ref{lem:structure}\ref{it:global-support} shows that the returned
certificate is feasible.

It remains to consider the gap case \eqref{eq:gap-case}, which gives
$L(S^*)<-\varepsilon c$.  If
$c\ge s(s-1)/(2\beta)$, Corollary~\ref{cor:singleton} gives a
positive-weight singleton of ratio at most $(1+\beta)\rho^*<1$, and family
(i) contains a negative-$L$ candidate.  If $0<c\le\tau s^2$, then \eqref{eq:parameters} gives
\(\tau\le1/200\) and \(64\tau<\varepsilon\).  Lemma~\ref{lem:near} returns a candidate $T$ with
\[
  L(T)\le L(S^*)+\frac{64c^2}{s^2}
       <-\varepsilon c+64\tau c<0,
\]
so family (ii) contains a negative-$L$ candidate.  In the remaining case,
\eqref{eq:intermediate-regime} holds.  For the correct guess $s=|S^*|$, the discussion following Lemma~\ref{lem:threshold-localization} shows that some $v\in S^*\cap L_s$ produces
a universe $U(v)$ satisfying~\eqref{eq:deterministic-localization-size}.  If some coefficient $b_u$ in~\eqref{eq:F-quadratic} is negative, recovery returns the violated singleton
$\{u\}$.  Otherwise the deterministic implementation of Lemma~\ref{lem:dense-quadratic}, applied with accuracy $\xi$ from~\eqref{eq:localized-accuracy}, returns a set of negative $L$-value by the argument in Section~\ref{sec:lagrangian}.  Hence family (iii) also contains a negative-$L$ candidate.

There are polynomially many loops over $s$ and pivots, and every threshold closure can be constructed in polynomial time.  Moreover, $1/\xi=\poly(1/\varepsilon)$ and $\langle b\rangle\le\langle q\rangle+O(n\log n)$.  Remark~\ref{rem:bit-complexity} therefore bounds every recovery call by $2^{\poly(1/\varepsilon)}(n+\langle q\rangle)^{O(1)}$ bit operations, while exact evaluation of $L(T)$ also takes polynomial time.  The total running time is
$2^{\poly(1/\varepsilon)}(n+\langle q\rangle)^{O(1)}$.
\end{proof}

\input{rounding_simplified}

\section*{Acknowledgments}
DGS is supported by the grant PID2024-155946NB-I00 funded by Ministerio de Ciencia, Innovación y Universidades (MICIU), Agencia Estatal de Investigación (AEI/10.13039/501100011033) and the European Social Fund Plus (ESF+).

We used OpenAI's ChatGPT as an implementation aid to assist in writing the code for the computer-assisted verifier of Lemma~\ref{lem:verification}, and
to conduct a preliminary numerical search for better rounding rules leading up to~\eqref{eq:analytic-rules}. We independently checked and tested all generated code for correctness.  We assume responsibility for all content.

\bibliographystyle{abbrvnat}
\bibliography{main}

\appendix

\input{extreme_regime_proofs}

\input{dense_quadratic_linear}
\input{rounding_verification}

\input{hardness}

\input{query_lb}
\section{Applications}\label{app:applications}
\input{seeded_cc}

\input{cluster_insertion}

\input{rounding_derandomization}
\input{covering_gap}

\end{document}

%% file: abstract.tex
We give a $(\ratioconst+\varepsilon)$-approximation for
correlation clustering on complete graphs, improving the previous best
factor of $1.485+\varepsilon$ of Cao et al.\ (STOC 2024).

Our first main contribution is an efficient approximate separation oracle for the
cluster-LP dual. Given signed vertex weights $q$, it either finds a set
$S$ with $q(S)>\cost(S)$ or certifies that $q/(1+\varepsilon)$ is dual
feasible,
where $\cost(S)$ measures the correlation clustering disagreements attributed to $S$ in any clustering
in which $S$ is a cluster.
Via the ellipsoid method, this yields $(1+\varepsilon)$-approximate
primal and dual solutions for the fractional cluster LP in deterministic time $2^{\poly(1/\varepsilon)}n^{O(1)}$. The separator works directly on the original instance,
without global preclustering: a novel localization argument restricts the search
to a small universe, where weak regularity handles the resulting dense
quadratic minimization.
Complementing this result, we prove that exact dual separation and cluster-LP optimization are NP-hard, even for complete unweighted instances.

Our second main contribution is a new rounding of the cluster LP. It
combines cluster-based rounding with a continuous conditional pivot rule,
and its analysis relies on a single variance inequality with explicit
weights. This bounds the cluster-LP integrality gap by
$\ratioconst$, near its known lower bound of $4/3$, and gives a
per-instance primal--dual certificate of approximation.

Finally, we extend our separator to obtain a $(1.92+\varepsilon)$-approximation for seeded correlation clustering (where each cluster may contain at most one seed from a prescribed
        seed
        set), and to
bounded-weight instances, yielding a deterministic polynomial-time
approximate implementation of the cluster-insertion primitive used in
combinatorial correlation clustering.

%% file: localization_threshold.tex
The goal of this section is to
generate a family of candidate universes, one of
which has size $O_\beta(s)$
and contains
the unknown optimum $S^*$.

First, we abstract away the key half-density property of $S^*$ expressed by Lemma~\ref{lem:structure}\ref{it:half-density}:
\begin{definition}[Cut-dense set]\label{def:cut-dense}
For $\delta\in(0,1]$, a set $S\subseteq V$ is \emph{$\delta$-cut-dense in $H$} if
\[
  e_H(A,S\setminus A)\ge \delta |A|\,|S\setminus A|
  \qquad\text{for every }A\subseteq S,
\]
i.e., if every cut of $H[S]$ contains at least a $\delta$ fraction of all possible edges across that cut.
\end{definition}

The algorithm makes a guess $s$ for the size of $S^*$ (which is $1/2$-cut-dense); one of them will be correct ($s = |S^*|$).
For each $s$, define the low-degree set of eligible vertices
\begin{equation*}
  L_s=\left\{u\in V:
  d(u)\le s-1+\frac{100s}{\beta}\right\}.
\end{equation*}
By Lemma~\ref{lem:structure}\ref{it:pivots},
   low-degree vertices comprise the majority of $S^*$ (and, in particular, $L_s \neq \emptyset$):
\begin{equation}|S^*\setminus L_s|\le\frac{2\beta\cost(S^*)}{100s}<\frac{s}{100}.
\end{equation}
Moreover, every $w\in L_s$ satisfies
\[
  |\Gamma^+(w)|=d(w)+1\le s+\frac{100s}{\beta}\le\frac{101s}{\beta}.
\]

\subsection{Localization}\label{sec:localization}
The next lemma treats $S$ as an unknown hidden set.  The algorithm knows its size $s$ and a set $P$ of \emph{eligible pivots}.  It never needs to recognize which elements of $P$ belong to $S$.

\begin{lemma}[Localization of a hidden cut-dense set]\label{lem:threshold-localization}
Let $H=(V,E)$ be a graph, let $S\subseteq V$ have size $s\ge2$, let $P\subseteq V$, and let $D\ge1$.  Suppose that, for some $\delta\in(0,1]$ and $\eta\in[0,\delta^2)$,
\begin{enumerate}[label=(\roman*),ref=(\roman*)]
  \item\label{it:threshold-cut-dense} $S$ is $\delta$-cut-dense in $H$;
  \item\label{it:threshold-exceptions} $|S\setminus P|\le\eta s$;
  \item\label{it:threshold-degree-cap} $|\Gamma^+_H(w)|\le Ds$ for every $w\in P$, where $\Gamma^+_H(w)$ denotes the closed neighborhood of $w$ in $H$.
\end{enumerate}
Choose a threshold $\theta$ with $0<\theta<\delta^2-\eta$, and define
\[
  \Phi_{s,\theta}(U)
  =U\cup\{u\in V:d^H_{U\cap P}(u)\ge\theta s\},
  \qquad
  \kappa=\frac{1-\delta^2+\eta}{1-\theta}<1.
\]
Let
\begin{equation}\label{eq:threshold-round-count}
  j=\min\left\{r\in\mathbb Z_{\ge0}:(1-\delta)\kappa^r\le\delta-\eta-\theta\right\}.
\end{equation}
For every $v\in P$, set
\[
  U_0(v)=\Gamma^+_H(v),
  \qquad
  U_{i+1}(v)=\Phi_{s,\theta}(U_i(v))\quad(0\le i\le j),
  \qquad
  U(v)=U_{j+1}(v).
\]
Then, for every $v\in S\cap P$,
\begin{equation}\label{eq:threshold-localization-conclusion}
  S\subseteq U(v)
  \qquad\text{and}\qquad
  |U(v)|\le D\left(1+\frac D\theta\right)^{j+1}s.
\end{equation}
Consequently, the family $\{U(v):v\in P\}$ contains a universe satisfying~\eqref{eq:threshold-localization-conclusion} for the unknown set $S$.
\end{lemma}

This can be viewed as a set-cover statement with a hidden ground set: the covering sets are the neighborhoods $\Gamma^+_H(v)$.  Ordinary high-frequency coverage would not suffice because the already exposed centers could lie in a part of $S$ having few edges to the rest.  Cut density rules out precisely this obstruction.  The operation $\Phi_{s,\theta}$ ensures that every vertex covered by sufficiently many neighborhoods of exposed eligible vertices is exposed in the next round.
In our main application, $S=S^*$, $s=|S^*|$ will be guessed, $\delta=1/2$, $P=L_s$, $D=101/\beta$, $\eta=1/100$ and $\theta=1/20$.

\begin{proof}
Fix $v\in S\cap P$, which exists by~\ref{it:threshold-exceptions}.  For brevity write $U_i=U_i(v)$, and put
\[
  W_i=U_i\cap P,\qquad Y_i=S\cap U_i,\qquad B_i=S\setminus U_i,\qquad A_i=S\cap W_i,
\]
and let $R=S\setminus P$ be the exceptional part of the hidden set.

\paragraph{The initial neighborhood covers a $\delta$ fraction of $S$.}
Applying cut density to the partition $S=\{v\}\mathbin{\dot\cup}(S\setminus\{v\})$ gives $d_S^H(v)\ge\delta(s-1)$.  Since the neighborhood includes $v$,
\begin{equation}\label{eq:threshold-initial-coverage}
  |Y_0|=|S\cap\Gamma^+_H(v)|\ge\delta(s-1)+1\ge\delta s,
  \qquad
  |B_0|\le(1-\delta)s.
\end{equation}

\paragraph{Every nonfinal closure contracts the hidden remainder.}
The sets $U_i$ only grow, so $|Y_i|\ge\delta s$.  Cut density across $Y_i\mathbin{\dot\cup}B_i$ and the fact that at most $|R||B_i|$ of these edges can have their endpoint in $Y_i\setminus P$ give
\begin{equation}\label{eq:threshold-contraction-cut}
  e_H(A_i,B_i)
  \ge\delta|Y_i||B_i|-|R||B_i|
  \ge(\delta^2-\eta)s|B_i|.
\end{equation}
Let $g_i=|B_i\cap U_{i+1}|$.  A vertex of $B_i\setminus U_{i+1}$ has fewer than $\theta s$ neighbors in $W_i$, and hence in $A_i\subseteq W_i$, while a vertex of $B_i\cap U_{i+1}$ has at most $s$ neighbors in $A_i$.  Therefore
\[
  (\delta^2-\eta)s|B_i|
  \le g_i s+(|B_i|-g_i)\theta s,
\]
and rearranging gives
\begin{equation}\label{eq:threshold-contraction}
  |B_{i+1}|\le\kappa|B_i|.
\end{equation}
After $j$ such contractions,
      $
  |B_j|\le(1-\delta)\kappa^j s\le(\delta-\eta-\theta)s.
  $

\paragraph{The final closure captures everything else.}
Fix $u\in B_j$.  The singleton consequence of cut density gives $d_S^H(u)\ge\delta(s-1)$.  All its neighbors in $S$ that need not belong to $A_j$ lie either in $B_j\setminus\{u\}$ or in $R$.  Hence
\[
\begin{aligned}
  d^H_{W_j}(u)
  \ge d^H_{A_j}(u)
  &\ge d^H_S(u)-(|B_j|-1)-|R|\\
  &\ge\delta(s-1)-(\delta-\eta-\theta)s+1-\eta s\\
  &=\theta s+1-\delta
  \ge\theta s.
\end{aligned}
\]
Thus $u$ is added to $U_{j+1}$.  This holds for every $u\in B_j$, proving $S\subseteq U_{j+1}$.

\paragraph{The visible universe remains small.}
For $0\le i\le j$, let
\[
  C_i=\{u\in V:d^H_{W_i}(u)\ge\theta s\}.
\]
Double-counting edge incidences with $W_i$, followed by~\ref{it:threshold-degree-cap}, gives
\[
  \theta s|C_i|
  \le\sum_{u\in V}d^H_{W_i}(u)
  =\sum_{w\in W_i}d^H(w)
  \le Ds|W_i|.
\]
Therefore $|C_i|\le(D/\theta)|W_i|\le(D/\theta)|U_i|$, and hence
\[
  |U_{i+1}|\le\left(1+\frac D\theta\right)|U_i|.
\]
Finally, $|U_0|=|\Gamma^+_H(v)|\le Ds$.  Applying the preceding growth bound
proves~\eqref{eq:threshold-localization-conclusion}.
\end{proof}

For our application, $\kappa=4/5$ and~\eqref{eq:threshold-round-count} gives $j=1$, so the construction performs exactly two threshold closures.  Abbreviate
\begin{equation}\label{eq:threshold-closure}
  \Phi_s(U):=\Phi_{s,1/20}(U)
  =U\cup\left\{u\in V:d^H_{U\cap P}(u)\ge\frac{s}{20}\right\},
\end{equation}
and, for every $v\in P$, form
\begin{equation}\label{eq:threshold-universe}
  U(v)=\Phi_s\bigl(\Phi_s(\Gamma^+_H(v))\bigr).
\end{equation}
Apply Lemma~\ref{lem:threshold-localization} to the positive-edge graph $H=(V,E^+)$ with $P=L_s$ and $D=101/\beta$.  For the correct guess $s=|S^*|$,
      Lemma~\ref{lem:structure}\ref{it:half-density} and the bounds preceding the localization lemma verify its hypotheses.  Since $\beta\le 1$, every $v\in S^*\cap L_s$ produces a universe satisfying
\begin{equation}\label{eq:deterministic-localization-size}
  S^*\subseteq U(v),
  \qquad
  |U(v)|\le \frac{101}{\beta}\left(1+\frac{2020}{\beta}\right)^2s
  \le\left(\frac{2021}{\beta}\right)^3s.
\end{equation}

\subsection{Recovery: extracting a negative witness}\label{sec:lagrangian}
It remains to minimize the Lagrangian $L$ on each candidate universe $U$.  For $T\subseteq U$, let $x=\one[T]$, and let $A$ be the symmetric zero-diagonal matrix with $A_{uv}=1-2\cdot\one[uv\in E^+]$ for $u\ne v$.  Expanding the definition of $L$ gives
\begin{equation}\label{eq:F-quadratic}
  L(x)=\frac12x^{\mathsf T}Ax+b^{\mathsf T}x,
  \qquad
  b_u=\frac12d(u)-q_u.
\end{equation}
If $b_u<0$ for some $u\in U$, then $L(\{u\})=b_u<0$, so this singleton is already a violated constraint.  We may therefore assume that $b\ge0$ on $U$.

The following lemma, proved in Appendix~\ref{sec:dense}, performs the required minimization.
\begin{lemma}\label{lem:dense-quadratic}
Let $A$ be a symmetric $n\times n$ matrix with zero diagonal and entries in
$[-1,1]$,
    and let
$b\in\mathbb Q_{\ge0}^n$.  For
\[
  \Phi(x)=\frac12x^{\mathsf T}Ax+b^{\mathsf T}x,
  \qquad x\in\{0,1\}^n,
\]
and every $\xi>0$, one can find $x\in\{0,1\}^n$ such that
\[
  \Phi(x)\le\min_{y\in\{0,1\}^n}\Phi(y)+\xi n^2
\]
either
\begin{enumerate}[label=(\roman*)]
  \item with constant success probability, in time $2^{\poly(1/\xi)}+\poly(1/\xi)n$; or
  \item deterministically, in time $2^{\poly(1/\xi)}+\poly(1/\xi)n^2$.
\end{enumerate}
Both bounds are in the RAM model with unit-cost arithmetic and query access to the entries of~$A$.
\end{lemma}
To apply it,  set
\begin{equation}\label{eq:localized-accuracy}
  \xi=\frac{\varepsilon\tau}{2}\left(\frac{\beta}{2021}\right)^6.
\end{equation}
For the universe in~\eqref{eq:deterministic-localization-size}, the intermediate-regime lower bound $\cost(S^*)>\tau s^2$ gives
\[
  \xi|U|^2\le\frac{\varepsilon\tau s^2}{2}<\frac{\varepsilon\cost(S^*)}{2}.
\]
Because $S^*\subseteq U$ and the gap case gives $L(S^*)<-\varepsilon\cost(S^*)$, the set $T\subseteq U$ represented by $x$ in Lemma~\ref{lem:dense-quadratic} provides the desired
witness:
\[
 L(T)\le\min_{X\subseteq U}L(X)+\xi|U|^2
  \le L(S^*)+\xi|U|^2<-\frac{\varepsilon\cost(S^*)}{2}<0.
\]

%% file: rounding_simplified.tex
\section{Rounding the cluster LP}
\label{sec:rounding}
We now formally define and analyze the rounding scheme outlined in Section~\ref{sec:rounding-overview}. Let $(z_S)_{S\subseteq V}$ be any feasible solution of the cluster LP; we may assume $z_\emptyset=0$.

\begin{theorem}[Rounding the cluster LP]\label{thm:rounding}
There is a
polynomial-time algorithm which, given a
\mindisagree{} instance on a complete signed graph and a feasible
solution $z$ of~\eqref{eq:cluster-primal}, returns a
clustering $\mathcal A$ satisfying
\[
  \mathbb E[\cost(\mathcal A)]
  \le \alpha_0\sum_{S\subseteq V}\cost(S)z_S,\qquad \text{where } \qquad
  \alpha_0:=1.3865.
\]
Moreover, if $z$ is given explicitly and $z_S\ge\Delta$ for every $S$ in its support, the algorithm can be implemented in expected time $\widetilde O(n/\Delta^2)$ in the adjacency-list query model of Cao et al.~\cite{Cao+25} (degree, neighbor and pair queries, each answered in $O(1)$ time).
\end{theorem}
As in the rounding of Cao et al.~\cite[Section~2.2]{Cao+24arXiv}, the polynomial-time procedure can be derandomized by the method of pessimistic estimators of conditional
expectations~\cite{alon2016probabilistic}; see Appendix~\ref{app:rounding-derandomization}.  (This does not apply to the sublinear-time implementation in the second part.)

For distinct vertices put
\begin{equation}\label{eq:marginals}
  y_{uv}:=\sum_{S\supseteq\{u,v\}}z_S,
  \qquad
  y_{uvw}:=\sum_{S\supseteq\{u,v,w\}}z_S.
\end{equation}
Fixing $u$, the weights $(z_S)_{S\ni u}$ form a probability distribution.
Thus, if $S$ is sampled from this distribution and
$A_v:=\mathbf 1[v\in S]$, then
$ \mathbb E[A_v]=\Pr[A_v=1] = y_{uv}$ and $\mathbb E[A_vA_w]=y_{uvw}.$

Using $1\ge \Pr[X \vee Y] = \Pr[X] + \Pr[Y] - \Pr[X \wedge Y]$,
we get from~\eqref{eq:cluster-primal}
      the local feasibility conditions
   that every three distinct vertices $u,v,w$ need to satisfy:
\begin{equation}\label{eq:local-feas}
\max\{0,y_{uv}+y_{uw}-1,y_{uv}+y_{vw}-1,y_{uw}+y_{vw}-1\}
\le y_{uvw}\le\min\{y_{uv},y_{uw},y_{vw}\}.
\end{equation}
The LP contribution of a pair $uv$ with LP value $y = y_{uv}$ is
$
  \ell^+(y)=1-y,
  \ell^-(y)=y.
  $
  Indeed,
\[
  \sum_S\cost(S)z_S
  =\sum_{uv\in E^-}\underbrace{\sum_{S\supseteq\{u,v\}}z_S}_{y_{uv}}
   +\frac12\sum_{uv\in E^+}\underbrace{\sum_{S:\,|S\cap\{u,v\}|=1}z_S}_{2(1-y_{uv})}
  =\sum_{uv\in\binom V2}\ell^{\sigma(uv)}(y_{uv}),
\]
where we used $\sum_{S\ni u}z_S=\sum_{S\ni v}z_S=1$.

\subsection{The two rounding procedures}
We run the two rounding procedures below independently
and return the cheaper clustering.
\paragraph{Cluster-based rounding $\mathcal C$.}
Repeatedly perform the following: First, draw a set with probability proportional to $z_S$. Then, output
$S$ intersected with the currently unclustered
vertices as a new cluster, if this intersection is nonempty. Repeat until all vertices have been clustered.
For a pair $uv$ of LP value $y=y_{uv}$, the first sampled set intersecting ${u,v}$ contains both with probability $y/(1+1-y)$~\cite[Lemma~6]{Cao+24arXiv}.
Hence the expected disagreement cost, according as the edge is positive or negative, is, respectively,
\begin{equation}\label{eq:clock-cost}
  \kappa^+(y)=\frac{2(1-y)}{2-y} = \frac{2}{2-y} \ell^+(y),
  \qquad
  \kappa^-(y)=\frac{y}{2-y} = \frac{1}{2-y} \ell^-(y).
\end{equation}
The ratio between the expected disagreement cost and the LP cost is $\le 2$ (resp., $\le 1$) for positive (resp., negative) pairs.  Thus, cluster-based rounding performs particularly well on the latter.
Pivot rounding, described next, helps balance the tradeoff.

\paragraph{Conditional affine pivot rounding $\mathcal P$.}
In one round, choose a uniformly random alive (unclustered) pivot $u$, sample
$S\ni u$ according to $(z_S)_{S\ni u}$, and then independently add every other
alive vertex $v$ with probability
\begin{equation}\label{eq:conditional-affine}
  \pi(A_v,y_{uv}):=\mu_\sigma(y_{uv})+\gamma_\sigma(y_{uv})\,(A_v-y_{uv}),
  \qquad A_v=\mathbf 1[v\in S],
\end{equation}
where $\sigma=\sigma(uv)$ is the sign of the pair.  Here $A_v$ records whether
$v$ belongs to the LP cluster $S$ sampled from the pivot.  Since $\mathbb
E[A_v]=y_{uv}$, conditional on the pivot $u$ the probability that $v$ joins the
cluster of $u$ is $\mu_\sigma(y_{uv})$, while
$\gamma_\sigma$ controls how strongly the decision follows membership in $S$.

Equivalently, the rule is implemented with two coins: after sampling $S$, add
$v$ with probability $p_1(y):=\mu_\sigma(y)+\gamma_\sigma(y)(1-y)$ if $v\in S$,
and with probability $p_0(y):=\mu_\sigma(y)-\gamma_\sigma(y)\,y$ otherwise; the
functions below satisfy $0\le p_0, p_1\le1$, so this is well defined.

      The hybrid pivot family of~\cite[Algorithm~2]{Cao+24arXiv} becomes a special case: its independent mode has $\gamma=0$, while its dependent mode has $(\mu,\gamma)=(y,1)$.
Our framework is strictly more general: in the middle regime for positive edges, our rule specified below has $p_1(y)=1$ and $0<p_0(y)<1$, so it follows the sampled cluster while also
admitting outsiders.

\paragraph{Our rules.} We specify the closed-form rules that we use.  Write
\begin{equation}
  \boxed{
    \nu:=\frac{13}{6}\approx 2.17,
    \qquad
    q(y):=\Bigl[\nu\bigl(y-\tfrac12\bigr)\Bigr]_{[0,1]}
  }
\end{equation}
where $[\cdot]_{[0,1]}$ denotes truncation to $[0,1]$.  Then we set
\begin{equation}\label{eq:analytic-rules}
  \boxed{
    \begin{aligned}
      \mu_+(y)&=y+(1-y)q(y),
        &\qquad \gamma_+(y)&=1-q(y),\\
      \mu_-(y)&=\min\Bigl\{\tfrac32y,\;2y-y^2\Bigr\},
        &\qquad \gamma_-(y)&=0.
    \end{aligned}
  }
\end{equation}
For a positive pair this is exactly $p_1(y)=1$ and $p_0(y)=q(y)$:
\emph{every member of the sampled cluster joins, and an outsider is rescued
independently with probability $q(y)$}.  The rescue probability vanishes for
$y\le\frac12$, where the rule follows the sampled cluster exactly, and reaches
$1$ at
$y=\frac12+\frac1\nu=\frac{25}{26}\approx 0.962$,
    above which every vertex
joins.  For a negative pair the decision ignores $S$ and joins with probability
$\mu_-(y)$, which equals $\frac32y$ for $y\le\frac12$ and $2y-y^2$ for
$y\ge\frac12$.  The first expression lies in $[0,\frac34]$ and the second in $[\frac34,1]$.  Thus both conditional joining probabilities lie in $[0,1]$.

\subsection{Budget left for the pivot procedure}\label{sec:budgets}
We follow the residual-budget accounting of Cao et al.~\cite[Section~2.2]{Cao+24arXiv}.
What differs is the pivot rule and the certificate proving that these budgets cover the cost.

Since the algorithm returns the cheaper clustering from the two rounding procedures, we have
\[
  \mathbb E\bigl[\min\{\cost(\mathcal C),\cost(\mathcal P)\}\bigr]
  \le\omega\,\mathbb E[\cost(\mathcal C)]+(1-\omega)\,\mathbb E[\cost(\mathcal P)]
\]
for every
$\omega\in[0,1]$.
We take
    $
  \omega:=\frac{\alpha_0}{2}
  $
  as
an analysis-only parameter.

We aim for every pair to contribute at most $\alpha_0$ times its LP cost.
By~\eqref{eq:clock-cost} a pair of sign $\sigma$ and LP value $y$ contributes
exactly $\kappa^\sigma(y)$ to $\mathbb E[\cost(\mathcal C)]$, so the amount
$\omega\kappa^\sigma(y)$ of its target $\alpha_0\ell^\sigma(y)$ is already
spent.  Dividing what remains by the weight $1-\omega$ of the pivot procedure
leaves it the \emph{budget}
\begin{equation}\label{eq:budget}
  B^\sigma(y):=\frac{\alpha_0\ell^\sigma(y)-\omega\kappa^\sigma(y)}{1-\omega},
  \qquad\text{so that}\qquad
  \omega\kappa^\sigma(y)+(1-\omega)B^\sigma(y)=\alpha_0\ell^\sigma(y).
\end{equation}
Since $\omega=\alpha_0/2$, these are explicitly, with
$\lambda:=\frac{\alpha_0}{2-\alpha_0}$,
\begin{equation}\label{eq:budget-explicit}
  B^+(y)=2\lambda\,\frac{(1-y)^2}{2-y},
  \qquad
  B^-(y)=\lambda\,\frac{y(3-2y)}{2-y} = \lambda - B^+(y).
\end{equation}

For a pair $e$ we abbreviate $B_e:=B^{\sigma(e)}(y_e)$,
$\kappa_e:=\kappa^{\sigma(e)}(y_e)$ and $\ell_e:=\ell^{\sigma(e)}(y_e)$.  The
budget is a bookkeeping quantity, invisible to the algorithm: in the analysis it
is a credit attached to the pair $e$ and released in the one round in which the
fate of $e$ is decided.  The argument consists in showing that, round by
round, the credit released covers the cost incurred.

\subsection{Analysis}\label{sec:rounding_analysis}

Throughout this subsection we fix one round of the pivot procedure and
condition on the whole history of the algorithm before that round; let $W$ be
the alive set at its start.  Every expectation below is conditional on that
history, and $\mathbb E_u$ denotes the expectation conditioned in addition on
the pivot of the round being $u$.  The pivot is uniform on $W$: each
$u\in W$ occurs with probability $1/|W|$.

For $v\in W$ let
\[
  \mathsf X_v:=\mathbf 1[v\text{ joins the cluster built in this round}],
\]
so that $\mathsf X_u=1$ for the pivot $u$ itself.  Let $S$ be the LP cluster sampled from $u$ and, as before, put $A_v=\mathbf 1[v\in S]$.  For $v\in W\setminus\{u\}$, the
conditional-affine rule~\eqref{eq:conditional-affine}
gives, using $\mathbb E_u[A_v]=y_{uv}$,
\[
  \mathbb E_u[\mathsf X_v\mid S]
  =\mu_{\sigma(uv)}(y_{uv}) +\gamma_{\sigma(uv)}(y_{uv})(A_v-y_{uv}),
  \qquad A_v=\mathbf 1[v\in S],
\]
\begin{equation}\label{eq:pivot-marginal}
  M_{uv}:=\mathbb E_u[\mathsf X_v]
  =\mu_{\sigma(uv)}(y_{uv}).
\end{equation}

For distinct $v,w\in W\setminus\{u\}$,
$\mathsf X_v \text{ and } \mathsf X_w$
are independent conditional on $S$.
Recall that
$\mathbb E_u[A_v-y_{uv}] =
\mathbb E_u[A_w-y_{uw}] = 0$,
so
$\operatorname{Cov}_u[A_v, A_w] = \mathbb E_u[A_vA_w] - y_{uv} y_{uw} =  y_{uvw} - y_{uv}y_{uw}$.
Thus
\begin{align}
  J_u(v,w)
  :=\mathbb E_u[\mathsf X_v\mathsf X_w]
  &=\mathbb E_u[\mathbb E_u[\mathsf X_v\mathsf X_w\mid S]]
  =\mathbb E_u\!\left[
      \mathbb E_u[\mathsf X_v\mid S]\,
      \mathbb E_u[\mathsf X_w\mid S]
    \right] \notag\\
  &=
    M_{uv}M_{uw}
    +\gamma_{\sigma(uv)}(y_{uv})\,\gamma_{\sigma(uw)}(y_{uw}) \operatorname{Cov}_u[A_v, A_w] \notag\\
  &=M_{uv}M_{uw}
    +\gamma_{\sigma(uv)}(y_{uv})\,\gamma_{\sigma(uw)}(y_{uw})
      \bigl(y_{uvw}-y_{uv}y_{uw}\bigr).
      \label{eq:joint-unified}
\end{align}
Note that  $J_u(v, w)$ is affine in $y_{uvw}$. Moreover, for the rules in~\eqref{eq:analytic-rules}, $M_{uv}$ and $J_u(v, w)$ are piecewise polynomial in $y_{uv}$ and $y_{uw}$.

\paragraph{Decided edges, payments, and released budgets.}
Call a pair $e\in\binom W2$ \emph{decided} in this round if at least one of its
endpoints joins the cluster.
Such a pair leaves the alive set and is settled for good.
An undecided pair keeps both
endpoints alive. Every pair is decided in
exactly one round, and it contributes a disagreement to the final clustering
exactly when the random variable
\[
  \mathsf{Pay}(e):=
  \begin{cases}
    \mathbf 1[\mathsf X_v\neq \mathsf X_w], & e=vw\in E^+,\\
            \mathsf X_v\mathsf X_w,
    & e=vw\in E^-,
  \end{cases}
\]
equals $1$ in that round.  We attach to $e$ the credit $B_e$ of Section~\ref{sec:budgets}
and release it in the round where $e$ is decided; that is, we set
\[
  \mathsf{Rel}(e):=B_e\cdot\mathbf 1[\mathsf X_v=1 \vee \mathsf X_w=1],
  \qquad e=vw .
\]
Over the whole execution the cost of the clustering produced by $\mathcal P$ is
$\sum_{\text{rounds}}\sum_e\mathsf{Pay}(e)$, and the credit released is exactly
$\sum_{e\in\binom V2}B_e$.

These definitions also cover the pairs containing the pivot ($\mathsf X_u=1$): such a pair is decided with certainty.
Taking conditional expectations
using~\eqref{eq:pivot-marginal} and~\eqref{eq:joint-unified}, for distinct $v,w\in W\setminus\{u\}$,
\begin{align}
  \operatorname{pay}_u(vw)&:=\mathbb E_u[\mathsf{Pay}(vw)]
  =\begin{cases}
     M_{uv}+M_{uw}-2J_u(v,w), & vw\in E^+,\\
     J_u(v,w), & vw\in E^-,
   \end{cases}\label{eq:triangle-pay}\\
  \operatorname{rel}_u(vw)&:=\mathbb E_u[\mathsf{Rel}(vw)]
  =\bigl(M_{uv}+M_{uw}-J_u(v,w)\bigr)B_{vw},
\end{align}
whereas for a pair containing the pivot,
\begin{equation}\label{eq:degenerate-pay}
  \operatorname{pay}_u(uv)=
  \begin{cases}
    1-M_{uv}, & uv\in E^+,\\
    M_{uv}, & uv\in E^-,
  \end{cases}
  \qquad
  \operatorname{rel}_u(uv)=B_{uv}.
\end{equation}

\paragraph{One round in aggregate.}
For a triangle $T\in\binom W3$ and a pair $e\in\binom W2$ put
\[
  D(T):=\sum_{u\in T}
    \bigl(\operatorname{rel}_u(T\setminus\{u\})-\operatorname{pay}_u(T\setminus\{u\})\bigr),
  \qquad
  D(e):=\sum_{u\in e}
    \bigl(\operatorname{rel}_u(e)-\operatorname{pay}_u(e)\bigr).
\]
For readers familiar with the notation of~\cite{Cao+24arXiv}, their $\Delta(T)$ and $\operatorname{cost}(T)$ are respectively our total released budget and payment over the three choices of pivot, so $D(T)=\Delta(T)-\operatorname{cost}(T)$; the same correspondence holds for pairs, viewed there as degenerate triangles.

By~\eqref{eq:triangle-pay}--\eqref{eq:degenerate-pay}, each of these numbers is determined by the edge signs and LP values of $T$ and $e$.
In particular the
same triangle contributes the same $D(T)$ in every round in which its three
vertices are still alive; only the index sets in the sums below depend on
$W$.
Averaging over the
pivot,
\begin{equation}\label{eq:one-round-accounting}
  \mathbb E\Bigl[\sum_{e\in\binom W2}\bigl(\mathsf{Rel}(e)-\mathsf{Pay}(e)\bigr)\Bigr]
  =\frac1{|W|}\sum_{u\in W}\sum_{e\in\binom W2}
     \bigl(\operatorname{rel}_u(e)-\operatorname{pay}_u(e)\bigr)
  =\frac1{|W|}\Bigl(\sum_{T\in\binom W3}D(T)+\sum_{e\in\binom W2}D(e)\Bigr).
\end{equation}
For the second equality,
notice that each triangle collects one contribution to $D(T)$ for each of its three vertices in the role of a pivot, whereas each endpoint of an edge $e$ contributes once to $D(e)$.

We show that the quantity in~\eqref{eq:one-round-accounting} is nonnegative, i.e., the expected credit released dominates the expected cost incurred in every round. This would easily follow if
every $D(T)$ and every $D(e)$ were nonnegative, but unfortunately that is not the case (as already noted by~\cite{Cao+24}).
For instance, consider the all-positive triangle with
$y_{uv}=y_{uw}=y_{vw}=\frac12$ and $y_{uvw}=0$, which
satisfies~\eqref{eq:local-feas}.  Here $q(\frac12)=0$, so $\mu_+(\frac12)=\frac12$
and $\gamma_+(\frac12)=1$; hence $J_u(v,w)=\frac14+(0-\frac14)=0$ at each pivot,
the pair opposite the pivot is separated with probability $1$, and
$B^+(\frac12)=\lambda/3$.  Therefore
$  D(T)=3\Bigl(\frac{\lambda}{3}-1\Bigr)<0.
  $
The deficit of such a triangle must be compensated by other triangles or edges.

\paragraph{Variance certificates.}
We seek auxiliary quantities $V(T)$ and $V(e)$ with two properties: they are pointwise lower bounds on $D(T)$ and $D(e)$, and their total over every alive set is nonnegative.  The
second property can be forced by bounding the total by an expected sum of squares.  Indeed,~\eqref{eq:joint-unified} shows that the dependence of $D(T)$ on the triple value $y_{uvw}$
occurs through the covariances $y_{uvw}-y_{uv}y_{uw}$.  This suggests using a weighted quadratic form in the centered membership indicators.  Choose a weight function $\Theta_\sigma(y)$ for each sign (used only in the analysis).
For $S\sim (z_S)_{S\ni u}$, put
\[ Z_u:=\sum_{v\in W\setminus\{u\}} \Theta_{\sigma(uv)}(y_{uv})\bigl(\mathbf 1[v\in S]-y_{uv}\bigr).  \]
The cross terms in the expansion of $\mathbb E_u[Z_u^2]$ are weighted covariances and are naturally indexed by triangles; the diagonal terms are weighted variances and are naturally indexed by pairs.  Although individual cross terms may be negative, their sum with the diagonal terms is nonnegative.  Accordingly, define for
each $T=\{u,v,w\}$,
\begin{align}
  V(T)&:=2\sum_{p\in T}\Theta_{\sigma(pq)}(y_{pq})\,\Theta_{\sigma(pr)}(y_{pr})
        \bigl(y_{uvw}-y_{pq}y_{pr}\bigr),
        \qquad \{q,r\}=T\setminus\{p\},
        \label{eq:triangle-variance}\\
  V(e)&:=2\,\Theta_{\sigma(e)}(y_e)^2\, y_e (1-y_e).
\end{align}

\begin{lemma}\label{lem:variance-certificate}
Suppose that, for every locally feasible signed triangle $T$ and every signed
pair $e$,
\begin{equation}\label{eq:local-certificate}
  D(T)\ge V(T),
  \qquad
  D(e)\ge V(e).
\end{equation}
Then
$
  \mathbb E[\cost(\mathcal P)]\le\sum_eB_e.
$
Consequently, the cheaper of $\mathcal C$ and $\mathcal P$ has expected cost at
most $\alpha_0\sum_e\ell_e$.
\end{lemma}

\begin{proof}
Fix a round, condition on the history before it, and let $W$ be its alive set.
Expanding the square and using $\mathbb E[\mathbf 1[v\in S]]=y_{uv}$ and
$\mathbb E[\mathbf 1[v\in S]\mathbf 1[w\in S]]=y_{uvw}$, the cross terms
($v\neq w$) of $\sum_{u\in W}\mathbb E_u[Z_u^2]$ group into triangles and the
diagonal terms ($v=w$) into pairs:
\begin{equation}\label{eq:variance-decomposition}
  \sum_{u\in W}\mathbb E_u[Z_u^2]
  =\sum_{T\in\binom W3}V(T)+\sum_{e\in\binom W2}V(e)\ \ge\ 0 .
\end{equation}
Summing~\eqref{eq:local-certificate} over the triangles and pairs inside $W$
and combining with~\eqref{eq:variance-decomposition} gives
$\sum_TD(T)+\sum_eD(e)\ge0$.  By~\eqref{eq:one-round-accounting}, the
conditional expected credit released in the round is therefore at least the
conditional expected cost incurred in it.
This holds for each of at most $n$ rounds and every history, so taking expectations and
summing,
\[
  \mathbb E[\cost(\mathcal P)]
  =\mathbb E\Bigl[\sum_{\text{rounds}}\sum_e\mathsf{Pay}(e)\Bigr]
  \le\mathbb E\Bigl[\sum_{\text{rounds}}\sum_e\mathsf{Rel}(e)\Bigr]
  =\sum_{e\in\binom V2}B_e,
\]
the two equalities holding because each pair is decided in exactly one round,
in which it incurs its disagreement, if any, and releases its credit $B_e$.
Using
\eqref{eq:clock-cost} and~\eqref{eq:budget}, we finally get
\[
  \mathbb E\bigl[\min\{\cost(\mathcal C),\cost(\mathcal P)\}\bigr]
  \le\omega\sum_e\kappa_e+(1-\omega)\sum_eB_e
  \le\alpha_0\sum_e\ell_e. \qedhere
\]
\end{proof}

\subsection{The explicit certificate and its verification}\label{sec:rounding-check}

It remains to exhibit weights $\Theta_\pm$ for which~\eqref{eq:local-certificate}
holds.
Put
\begin{equation}\label{eq:theta-closed-form}
\boxed{
  \Theta_+(y):=\sqrt{\rho}\;y(1-y)^2,
  \qquad
  \Theta_-(y):=0,
  \qquad\text{with}\qquad
  \rho:=\frac{1012}{25} = 40.48.
}
\end{equation}
Thus, for positive edges, $\Theta_+(uv)$ is proportional to $\ell^+(uv)\, \operatorname{Var}_u[A_v]$.
For negative edges,
$\Theta_-\equiv0$
means that the
triple value $y_{uvw}$ influences $D(T)$ only through pivots both of whose
incident edges are positive.
The statistic of Section~\ref{sec:rounding_analysis} becomes then
\[ Z_u=\sqrt{\rho}\sum_{\substack{v\in W\setminus\{u\}\\ uv\in E^+}} y_{uv}(1-y_{uv})^2\bigl(\mathbf 1[v\in S]-y_{uv}\bigr),  \]
and the inequalities in~\eqref{eq:local-certificate} to be checked become
\begin{align*}
  D(e)&\ge
  \begin{cases}
    2\rho\,y_e^3(1-y_e)^5,&e\in E^+,\\
    0,&e\in E^-,
  \end{cases}\\
  D(T)&\ge
  2\rho
  \sum_{\substack{p\in T:\,pq,pr\in E^+\\
  \{q,r\}=T\setminus\{p\}}}
  y_{pq}y_{pr}(1-y_{pq})^2(1-y_{pr})^2
  \bigl(y_{uvw}-y_{pq}y_{pr}\bigr).
\end{align*}
For our rules and budgets, $D(T)$ and $D(e)$ are piecewise rational functions with positive denominators of degrees 3 and 1,
    respectively. Moreover, $D(uvw)$ is affine in the triple value $y_{uvw}$.

\begin{lemma}\label{lem:verification}
For $\alpha_0=\frac{2773}{2000}$, the rules~\eqref{eq:analytic-rules}, the
exact budgets~\eqref{eq:budget-explicit} and the
weights~\eqref{eq:theta-closed-form}, all inequalities
of~\eqref{eq:local-certificate} hold.
\end{lemma}

These inequalities reduce to finitely many polynomial nonnegativity checks in the pair values:
$D(uvw)-V(uvw)$ is affine in the triple value $y_{uvw}$, so only the endpoints of~\eqref{eq:local-feas} matter, and
    clearing positive denominators leaves polynomials in $y_{uv}, y_{uw}, y_{vw}$.
    We verify them using exact
rational arithmetic.
For each triangle sign pattern, we split the feasible pair-value domain
into regions on which the rounding rules have fixed polynomial formulas;
  these conditions define rational polytopes in
  $(y_{uv},y_{uw},y_{vw})$-space.
  The verifier triangulates these polytopes into simplices; on each resulting simplex, all coefficients of the polynomial obtained from $D(T) - V(T)$ by clearing denominators, written in the Bernstein basis, are nonnegative.
Since the Bernstein basis functions are nonnegative and sum to one, this certifies the inequalities throughout their domains. Subdivision is necessary
because a polynomial may be positive even if one of its Bernstein
coefficients is negative; when this happens  (as it did in 3 out of the 35 triangle cases), we bisect the simplex at the midpoint of its longest edge, repeating until nonnegativity
is certified separately within each piece.
A detailed explanation of the full computer-assisted proof of Lemma~\ref{lem:verification} appears in Appendix~\ref{app:rounding-verification}.

\begin{proof}[Proof of Theorem~\ref{thm:rounding}]
Lemma~\ref{lem:verification} supplies the local inequalities required by
Lemma~\ref{lem:variance-certificate}, giving
\[
  \mathbb E[\cost(\mathcal A)]
  \le\alpha_0\sum_e\ell_e
  =\alpha_0\sum_S\cost(S)z_S.
\]
Both procedures run in time polynomial in $n$ and the support size of $z$.

Finally, when the rounding is combined with the adjacency-list algorithm of~\citet{Cao+25}, we need a sharper sublinear-time implementation.
To avoid computing the clustering costs, run cluster-based rounding with probability $\omega=\alpha_0/2$ and pivot rounding with probability $1-\omega$.
Cluster-based rounding is identical to that of~\citet{Cao+25} and can be implemented in time $\widetilde O(n/\Delta)$. As for our pivot-based rounding variant, with pivot $u$ and alive set
$W$, put
$
  U_u=W\cap\bigcup_{S\ni u:\,z_S>0}S.
  $
If $v\notin U_u$, then $y_{uv}=0$ and the rule never adds $v$ to the cluster of pivot $u$. If $v\in U_u$, then $y_{uv}\ge\Delta$, and~\eqref{eq:analytic-rules} gives
$
  \mu_+(y)\ge y$ and $  \mu_-(y)=\min\{\tfrac32y,2y-y^2\}\ge y.$
Thus a round with pivot $u$ removes at least $\Delta|U_u|$ vertices in expectation. Since the covering constraints imply that every vertex belongs to at most $1/\Delta$ sets in the support of $z$, the sets containing $u$, the sampled set, all values $y_{uv}$, and all joining decisions (each using one pair query, for the sign of $uv$) can be processed in time $O(|U_u|/\Delta)$. The amortization of~\citet[Section~9.2]{Cao+25} therefore gives expected total time $\widetilde O(n/\Delta^2)$ for pivot rounding, and hence for the randomized choice between the two procedures.
\end{proof}

%% file: extreme_regime_proofs.tex
\section{Proofs for the large and small cost regimes}\label{app:extreme-regime-proofs}
We give the proofs deferred from Section~\ref{sec:extreme-regimes}.

\begin{proof}[Proof of Lemma~\ref{lem:structure}]
\phantom{ nothing }
\begin{enumerate}[label=(\alph*)]
\item  This is essentially the definition of $\rho^*$ in Problem~\ref{prob:cc-minratio}, except that the inequality in~\eqref{eq:global-support} is also required to hold when $q(T) \le
0$. But in that case,
$\rho^*q(T)\le0\le\cost(T)$.
\item This is a consequence of \eqref{eq:pseudo-boolean}. We have
    \begin{align*}
        \cost(A \cup B)     &= \binom{|A \cup B|}{2} + \frac12\sum_{u \in A \cup B} d(u) - 2e^+(A \cup B),\\
        \cost(A) + \cost(B) &= \binom{|A|}{2} + \binom{|B|}{2} + \frac12\sum_{u \in A \cup B} d(u) -2e^+(A)-2e^+(B).
    \end{align*}
    Using  $\binom{a + b}{2} = \binom{a}{2} + \binom{b}{2} + ab$, we get
    \begin{align*}
        \cost(A) + \cost(B) - \cost(A \cup B) &= -|A| |B| + 2(e^+(A \cup B) - e^+(A) - e^+(B)) \\&= -|A||B| + 2e^+(A,B).
    \end{align*}

\item For a partition $S^*=A\mathbin{\dot\cup}B$, \eqref{eq:global-support} and \eqref{eq:partition-identity} give
\[
  \cost(S^*)+2e^+(A,B)-|A||B|=\cost(A)+\cost(B)\ge\rho^*(q(A)+q(B))=\cost(S^*).
\]

\item Write $D=T\setminus X$.  By \eqref{eq:partition-identity}, nonnegativity of $\cost(D)$, and $e^+(D,X)\le|D||X|$,
\[
  \cost(X)=\cost(T)-\cost(D)+2e^+(D,X)-|D||X|\le\cost(T)+|D||X|.
\]

\item
For $v\in T$,
write
$\delta_T(v):=|\Gamma^+(v)\mathbin\triangle T|=a_v+b_v$, where
\[
  a_v=|T\setminus\Gamma^+(v)|=|T|-1-d_T(v),
  \qquad
  b_v=|\Gamma^+(v)\setminus T|=d_{V\setminus T}(v).
\]
Every internal negative pair contributes to $a_v$ at both endpoints, while every positive pair crossing from $T$ to $V\setminus T$ contributes to $b_v$ at its endpoint in $T$.  Therefore
\[
  \sum_{v\in T}\delta_T(v)
  =2e^-(T)+e^+(T,V\setminus T)
  =2\cost(T),
\]
which proves \eqref{eq:pivot-sum} and the last bound in~\eqref{eq:dvs}. Since $d(v)=d_T(v)+d_{V\setminus T}(v) =(|T|-1-a_v)+b_v$,
\[
  |d(v)-(|T|-1)|=|b_v-a_v|\le a_v+b_v=\delta_T(v).
\]
Summing the last inequality proves the first bound in~\eqref{eq:dvs}.

\item
For every $u\in T$ with $q_u>0$, part~\ref{it:global-support} applied to $\{u\}$ gives
\[
  \rho^*q_u\le\cost(\{u\})=\frac{d(u)}{2}.
\]
Summing only over positive coordinates gives
\[
  \rho^*q^+(T)\le\frac12\sum_{\substack{u\in T\\q_u>0}}d(u)\le\frac12\sum_{u\in T}d(u)=e^+(T)+\frac12e^+(T,V\setminus T)\le\binom{|T|}{2}+\cost(T).
\]
Taking $T=S^*$ and dividing by $\cost(S^*)=\rho^*q(S^*)>0$ proves the final claim.
\end{enumerate}
\end{proof}

\begin{proof}[Proof of Corollary~\ref{cor:singleton}]
Let $P=\{u\in S^*:q_u>0\}$.  Weighted averaging, $q^+(S^*)\ge q(S^*)>0$, and \eqref{eq:positive-mass} give
\[
  \min_{u\in P}\rho(\{u\})=\min_{u\in P}\frac{d(u)}{2q_u}
  \le\frac{\frac12\sum_{u\in P}d(u)}{q^+(S^*)}
  \le\frac{\binom{|S^*|}{2}+\cost(S^*)}{q(S^*)}
  =\rho^*\left(1+\frac{|S^*|(|S^*|-1)}{2\cost(S^*)}\right).
\]
\end{proof}

\begin{proof}[Proof of Lemma~\ref{lem:fixed-cardinality}]
Put $X=U\setminus D$, so $|X|=t$.  Equation~\eqref{eq:pseudo-boolean} gives
\[
  L(X)=\binom t2+\frac12\sum_{u\in X}d(u)-2e^+(X)-q(X).
\]
Moreover,
\[
  \sum_{u\in X}d(u)=\sum_{u\in U}d(u)-\sum_{u\in D}d(u),
  \qquad
  q(X)=q(U)-q(D),
\]
and
\[
  e^+(X)=e^+(U)-\sum_{u\in D}d_U(u)+e^+(D).
\]

For the last identity, edges between $D$ and $X$ are counted once by
$\sum_{u\in D}d_U(u)$,
    while edges inside $D$ are counted twice.  Substituting these three
identities proves \eqref{eq:fixed-cardinality}, with
\[ K_{U,t}=\binom t2+\frac12\sum_{u\in U}d(u)-2e^+(U)-q(U). \]
\end{proof}

\begin{proof}[Proof of Lemma~\ref{lem:near}]
By Lemma~\ref{lem:structure}\ref{it:pivots}, choose $v\in S^*$
such that
\[
  |\Gamma^+(v)\mathbin\triangle S^*|
  \le\frac{2\cost(S^*)}{|S^*|}
  \le\frac{|S^*|}{100}.
\]
Put $A=\Gamma^+(v)$ and form $U$ using \eqref{eq:near-U} with $t=|S^*|$.

For every $u\in S^*$, Lemma~\ref{lem:structure}\ref{it:half-density} gives
\[
  d_A(u)\ge d_{S^*}(u)-|S^*\setminus A|
  \ge\frac{|S^*|-1}{2}-\frac{|S^*|}{100}
  \ge\frac{|S^*|}{3}.
\]
(because
$0<\cost(S^*)\le|S^*|^2/200 $ implies
 $|S^*|\ge10$).
Thus $S^*\subseteq U$.  Conversely, every $u\in U\setminus S^*$ satisfies, by \eqref{eq:near-U},
\[
  d_{S^*}(u)\ge d_A(u)-|A\setminus S^*|
  \ge\frac{|S^*|}{3}-\frac{|S^*|}{100}
  \ge\frac{|S^*|}{4}.
\]
These positive edges cross the cut of $S^*$, so
\[
  |U\setminus S^*|\,\frac{|S^*|}{4}
  \le e^+(S^*,V\setminus S^*)
  \le2\cost(S^*).
\]
In particular,
\begin{equation}\label{eq:near-extra-vertices}
  |U\setminus S^*|\le\frac{8\cost(S^*)}{|S^*|}.
\end{equation}

Recall that both $S^*$ and $T$ have the same size and, by design, $T$ minimizes the linear part of \eqref{eq:fixed-cardinality} among all subsets of $U$ of size $|S^*|$.
Therefore, by~\eqref{eq:near-extra-vertices},
\[
  L(T)\le L(S^*)+2e^+(U\setminus S^*)
      \le L(S^*)+|U\setminus S^*|^2
      \le L(S^*)+\frac{64\cost(S^*)^2}{|S^*|^2}.
\]
\end{proof}

%% file: dense_quadratic_linear.tex
\section{A dense quadratic subroutine}\label{sec:dense}
This appendix proves Lemma~\ref{lem:dense-quadratic}.
We follow the weak-regularity approach used by Bonchi, Garc\'ia-Soriano, and Kutzkov~\cite[Section 5]{local_corr} for local correlation clustering.  We use either the fast randomized cut decomposition of Frieze and Kannan~\cite{frieze_kannan_quick} or the deterministic version of Fox, Lov\'asz, and Zhao~\cite{FLZ19}; everything after the construction of the decomposition is the same.
Intuitively, two vertices in the same class of a weakly-regular partition have roughly the same number of connections with other classes and can thus be treated as equivalent.
Therefore, up to a small error, the quadratic part of $\Phi$ depends only on how many vertices are selected from each class, not on which ones; the linear term, which
varies within a class, then favors the vertices with the smallest coefficients $b_u$.
The partition has $2^{\poly(1/\xi)}$ classes, described by $\poly(1/\xi)$ weighted cuts.  Our algorithm works directly with the cut decomposition.
We discretize the intersection sizes of a solution with each cut set and
            group together vertices that are indistinguishable to
both the decomposition and a coarsened linear objective.  Every resulting LP
then has size depending only on $\xi$.

For a generic $N\times N$ matrix $A'$ with entries in
$[-1,1]$ and $R,S\subseteq[N]$, write
\[
  A'(R,S)=\sum_{i\in R}\sum_{j\in S}A'_{ij}=
\one_{R}A'\one_{S}^{\mathsf T}.
\]
Given an accuracy $\delta>0$, both constructions produce a decomposition
\begin{equation}
  \widetilde A
  =\sum_{r=1}^w d_r\one_{R_r}\one_{S_r}^{\mathsf T},
  \qquad w=\poly(1/\delta),
\end{equation}
such that, for every $R,S\subseteq[N]$,
\begin{equation}\label{eq:relative-cut-error}
  \left|A'(R,S)-\sum_{r=1}^w
  d_r|R\cap R_r|\,|S\cap S_r|\right|
  \le\delta N\sqrt{|R||S|}
  .
\end{equation}

For the randomized implementation, Theorem~1 and Section~4.2 of~\citet{frieze_kannan_quick} give~\eqref{eq:relative-cut-error}, with $w=O(\delta^{-4})$, in $\widetilde O(\delta^{-12})$ time and matrix queries, with constant success probability.  The cut sets are represented implicitly, and all their memberships can be evaluated in $N\poly(1/\delta)$ additional time.  Moreover, equation~(42) in that paper, together with $|R_r|,|S_r|\ge N/3$ and $\lVert A'\rVert_F^2\le N^2$, gives $\sum_r d_r^2\le27$.

For the deterministic implementation, Theorem~2.1 of~\citet{FLZ19} runs in $\poly(1/\delta)N^2$ time, outputs the cut sets explicitly, and has $w=O(\delta^{-16})$ and $|d_r|\le\delta^8/300$.  It guarantees $\lVert A'-\widetilde A\rVert_2\le\delta N$, from which~\eqref{eq:relative-cut-error} follows because
\[
  \left|\one_R^{\mathsf T}(A'-\widetilde A)\one_S\right|
  \le\lVert A'-\widetilde A\rVert_2\sqrt{|R||S|}.
\]
Thus, in the randomized and deterministic implementations respectively,
\begin{equation}\label{eq:coefficient-length}
  L:=\sum_{r=1}^w|d_r|\le\sqrt{27w}
  \qquad\text{and}\qquad
  L\le\frac{w\delta^8}{300}.
\end{equation}
In particular, for $0<\delta\le1$, both implementations have $w,L=\poly(1/\delta)$ and $L\le6w$.

The assertion of Lemma~\ref{lem:dense-quadratic} is immediate for $\xi\ge1$, so assume $0<\xi<1$.
By decreasing $\xi$ by at most a factor of two, we may further assume that it is dyadic and has encoding length $O(\log(1/\xi))$; we rename the resulting accuracy parameter $\xi$.
\paragraph{Removing large linear coefficients.}
If we fix all $x_u$ but one, $\Phi$ becomes linear in the remaining coordinate because $A$ is zero on the diagonal.  It follows that
no optimum contains a coordinate $u$ with $b_u\ge n$.  Indeed, if $x_u=1$
and $x'=x-e_u$, then
$
  \Phi(x')-\Phi(x)
  =-b_u-\sum_{v\ne u}A_{uv}x_v
    <0.
    $
Delete these coordinates and let $V_0$ be the remaining set, of size
$n_0=|V_0|$.
The optimum is unchanged;
if $n_0=0$, then
return $x=0$.

\paragraph{The cut surrogate.}
Apply the decomposition above to $A[V_0,V_0]$ with
    $
  \delta=\frac{\xi}{4}.
  $
Let $\mathcal H$ consist of all row and column sets $R_r,S_r$, let
$D=|\mathcal H|\le2w$, and put
$
  L=\sum_{r=1}^w|d_r|\le6w.
  $
For $X\subseteq V_0$, define
\[
  z_H(X)=|X\cap H|\qquad(H\in\mathcal H),
  \qquad
  P(z(X))=\frac12\sum_{r=1}^w
  d_rz_{R_r}(X)z_{S_r}(X).
\]
Taking $R=S=X$ in~\eqref{eq:relative-cut-error} gives the uniform estimate
\begin{equation}\label{eq:quadratic-cut-error}
  \left|\frac12\one_X^{\mathsf T}A\one_X-P(z(X))\right|
  \le\frac{\delta n_0^2}{2}.
\end{equation}
Moreover, for all $z,z'\in[0,n_0]^{\mathcal H}$,
\begin{equation}\label{eq:P-Lipschitz}
  |P(z)-P(z')|\le Ln_0\|z-z'\|_\infty,
\end{equation}
because $|ab-a'b'|\le n_0|a-a'|+n_0|b-b'|$ for
$a,b,a',b'\in[0,n_0]$.

Thus $P(z(X))$ approximates the quadratic term $\frac12\one_X^{\mathsf T}A\one_X$ in $\Phi(\one_X)$.  To minimize the full objective, we enumerate small boxes for the counting vector $z(X)$.  Since $P$ depends only on $z(X)$,~\eqref{eq:P-Lipschitz} makes it nearly constant within a box.
The linear term behaves differently:
two sets with the same cut counts may select vertices with very different coefficients $b_u$.  We therefore minimize the linear term subject to the box constraints; the box
containing an optimum of $\Phi$ then yields a near-optimal candidate.  To make the size of these LPs independent of $n_0$, we first coarsen $b$: vertices with the same cut-set
memberships and the same rounded coefficient become interchangeable.

\paragraph{Coarsening and aggregating the linear term.}
Set
\[
  h=\frac{\xi n}{8},
  \qquad
  \overline b_u=h\left\lceil\frac{b_u}{h}\right\rceil
  \quad(u\in V_0).
\]
Thus
\begin{equation}\label{eq:linear-rounding-error}
  0\le\overline b^{\mathsf T}\one_X-b^{\mathsf T}\one_X\le h|X|
  \le hn_0
  \qquad(X\subseteq V_0).
\end{equation}
As $b_u<n$, the integer $\overline b_u/h$ has
$O(1/\xi)$ possible values.

Place two vertices in the same class if they have the same membership vector
in the sets of $\mathcal H$ and the same value of $\overline b_u/h$.  Let
$\mathcal C$ be the collection of nonempty classes, and write
$\overline b_C=h j_C$ for the common rounded coefficient of class $C$.  Then
\begin{equation}
  |\mathcal C|
  \le2^D\left(2+\left\lceil\frac8\xi\right\rceil\right)
  =2^{\poly(1/\xi)}.
\end{equation}
The classes and their vertex lists are obtained in one scan of $V_0$.

\paragraph{Parameter-sized box LPs.}
Put
$
  \eta=\frac{\xi}{4(1+6w)}.
$
Since $L\le6w\le1+6w$, we have $L\eta\le\xi/4$.
If $n_0<4D/\eta$, exhaustive search takes $2^{\poly(1/\xi)}$ time, so assume
$n_0\ge4D/\eta$. Enumerate the axis-parallel grid
boxes of side length $\Delta=\eta n_0/2$ in $[0,n_0]^D$.  There are
$
  (O(1/\eta))^D=2^{\poly(1/\xi)}
  $
such boxes.  For each box $\prod_{H\in\mathcal H}[t_H,t_H+\Delta]$, where $t=(t_H)_{H\in\mathcal H}$ is its lower-corner vector, solve
{\small
\begin{equation}\label{eq:aggregated-statistic-LP}
\begin{aligned}
  \text{minimize}\quad
    &\sum_{C\in\mathcal C}j_Cy_C,\\
  \text{subject to}\quad
    &0\le y_C\le |C| &&(C\in\mathcal C),\\
    &t_H\le\sum_{\substack{C\in\mathcal C\\C\subseteq H}}y_C
       \le t_H+\Delta &&(H\in\mathcal H).
\end{aligned}
\end{equation}
}

In this aggregated LP, $y_C$ represents how many vertices of class $C$ are selected.  By definition, for every $C\in\mathcal C$ and $H\in\mathcal H$ either $C\subseteq H$ or $C\cap H=\varnothing$.  Hence $\sum_{C\subseteq H}y_C$ is the number of selected vertices in $H$, i.e., the $H$-coordinate of the counting vector; the last constraints require this vector to lie in the current box.  The LP has $2^{\poly(1/\xi)}$ variables and constraints, independently of $n$.

If the LP is feasible, take an optimal basic solution and put $\widehat y_C=\lfloor y_C\rfloor$.  At most $2D$ variables can be strictly between their bounds; otherwise the active aggregate constraints cannot determine a basic solution.
Rounding down cannot increase the LP objective and decreases each $H$-coordinate by at most
    $
  0\le\sum_{\substack{C\in\mathcal C\\C\subseteq H}}(y_C-\widehat y_C)<2D\le\Delta,$
because $n_0\ge4D/\eta$ and $\Delta=\eta n_0/2$.
The resulting integral class counts define a set $Y\subseteq V_0$ by taking
the prescribed number of arbitrary vertices from each class.  At this stage we
store only the count vector, rather than materializing $Y$.

For every candidate $Y$, compute only the surrogate
\begin{equation}
  \Psi(Y)=P(z(Y))+\overline b^{\mathsf T}\one_Y,
\end{equation}
which is determined by its class counts.  Return the candidate minimizing
$\Psi$; in particular, do not query all entries of $A[Y,Y]$.  Only after the
best count vector has been found do we materialize $Y$.

\begin{proof}[Proof of Lemma~\ref{lem:dense-quadratic}]
We are ready to analyze the error and runtime bounds of the scheme above:
\paragraph{Approximation guarantee.}
Let $X^\star$ be an optimum solution, which by the first paragraph is
contained in $V_0$, and consider the box containing $z(X^\star)$.  Its
aggregate count vector is feasible for~\eqref{eq:aggregated-statistic-LP}.
If $Y_0$ is the rounded candidate obtained from an optimal basic solution in
this box, then
\[
  \overline b^{\mathsf T}\one_{Y_0}
  \le\overline b^{\mathsf T}\one_{X^\star},
  \qquad
  \|z(Y_0)-z(X^\star)\|_\infty\le2\Delta=\eta n_0.
\]
Let $Y$ be the candidate returned.  Using
\eqref{eq:quadratic-cut-error}, \eqref{eq:P-Lipschitz},
\eqref{eq:linear-rounding-error}, and $\Psi(Y)\le\Psi(Y_0)$, we obtain
\[
\begin{aligned}
  \Phi(\one_Y)
  &\le\Psi(Y)+\frac{\delta n_0^2}{2}
  \le\Psi(Y_0)+\frac{\delta n_0^2}{2}\\
  &\le P(z(X^\star))+\overline b^{\mathsf T}\one_{X^\star}
       +L\eta n_0^2+\frac{\delta n_0^2}{2} &\\
  &\le\Phi(\one_{X^\star})
       +(\delta+L\eta)n_0^2+hn_0
  \le\Phi(\one_{X^\star})+\xi n^2.
\end{aligned}
\]
The last inequality uses
$\delta=\xi/4$, $L\eta\le\xi/4$, $h=\xi n/8$, and $n_0\le n$.

\paragraph{Running time.}
After the cut decomposition has been found, evaluating all cut signatures, bucketing the linear coefficients, constructing the
nonempty classes, and materializing the answer costs
$n\poly(1/\xi)$ arithmetic operations.  After coarsening, $b$ enters the aggregated LPs only through the integers $j_C\in\{0,\ldots,\lceil8/\xi\rceil\}$.
There are $2^{\poly(1/\xi)}$ boxes, and each aggregated LP has $2^{\poly(1/\xi)}$ variables and constraints, with a constraint matrix whose entries are in $\{0,\pm1\}$ (after introducing slack variables).
Tardos's algorithm~\cite{tardos_lp} solves a linear program with a number of arithmetic operations polynomial in its dimensions and in the encoding length of its constraint matrix,
    independently of the objective and the right-hand side; when the feasible region is a polytope, as here, the solution it returns is a vertex.
Hence all the LPs are solved with $2^{\poly(1/\xi)}$ arithmetic operations.  Using the randomized Frieze--Kannan construction gives total time $2^{\poly(1/\xi)}+n\poly(1/\xi)$ with constant success probability, whereas the deterministic Fox--Lov\'asz--Zhao construction gives $2^{\poly(1/\xi)}+n^2\poly(1/\xi)$ time.  This proves both parts of the lemma.
\end{proof}

\begin{remark}[Bit complexity]\label{rem:bit-complexity}
For the application in Section~\ref{sec:lagrangian}, $A$ has entries in $\{0,\pm1\}$.  The entries of $b$ are used in comparisons and in the rounding $\lceil b_u/h\rceil$, while the cut decompositions and the aggregated LPs use rational numbers of encoding length polynomial in $n+\langle b\rangle+\log(1/\xi)$.  Standard polynomial-time rational LP algorithms therefore implement the deterministic part of Lemma~\ref{lem:dense-quadratic} in
$
  2^{\poly(1/\xi)}(n+\langle b\rangle)^{O(1)}
  $
bit operations.
\end{remark}

%% file: rounding_verification.tex
\section{Computer-assisted verification of the rounding certificate}\label{app:rounding-verification}

This appendix gives the computer-assisted proof of Lemma~\ref{lem:verification}.

The following elementary property of the Bernstein basis on a
simplex~\cite[Section~8.2, p.~409]{bernstein_basis}
explains how our computer-assisted verifier certifies polynomial
inequalities over entire simplices.
\begin{fact}\label{fact:bernstein-enclosure}
Let $K\subseteq\mathbb R^k$ be a $k$-dimensional simplex with vertices $v_0,\ldots,v_k$.
For $x=(x_1,\ldots,x_k)\in K$, its $k+1$ barycentric coordinates $t_0(x),\ldots,t_k(x)$ are affine functions of $x$ satisfying
$x=\sum_{i=0}^k t_i(x)v_i$, $t_i(x)\ge0$, and $\sum_{i=0}^k t_i(x)=1$.
Write a $k$-variate polynomial $p(x_1,\ldots,x_k)$ of total degree at most $d$ in the degree-$d$ Bernstein basis on $K$:
\[
  p(x)=\sum_{|\beta|=d}b_\beta B_\beta(x),
  \qquad
  B_\beta(x)=\frac{d!}{\beta_0!\cdots\beta_k!}
             \prod_{i=0}^k t_i(x)^{\beta_i},
\]
where $\beta\in\mathbb Z_{\ge0}^{k+1}$ and $|\beta|=\sum_i\beta_i$.
Then
\[
  \min_{|\beta|=d}b_\beta\le p(x)\le\max_{|\beta|=d}b_\beta
  \qquad\text{for every }x\in K.
\]
In particular, nonnegativity of all coefficients $b_\beta$ certifies $p\ge0$ throughout $K$, including its boundary.
\end{fact}
\begin{proof}
The basis polynomials are nonnegative on $K$ and satisfy
$\sum_{|\beta|=d}B_\beta(x)=1$ by the multinomial theorem.
Thus $p(x)$ is a convex combination of the coefficients $b_\beta$.
\end{proof}

This is an exact change of basis.  If $p$ has rational coefficients and the vertices of $K$ are rational, all Bernstein coefficients are rational as well, so their signs can be checked in exact arithmetic.
A negative Bernstein coefficient makes the test inconclusive: the polynomial may still be nonnegative.
Restricting $p$ to smaller simplices and recomputing the coefficients can sharpen the bounds.
For a strictly positive polynomial, sufficiently fine subdivisions certify positivity.
In the finite subdivision completed by the verifier, all coefficients are nonnegative on every simplex.

For the pair checks we use $k=1$ in Fact~\ref{fact:bernstein-enclosure}.
On an interval $K=[a,b]$, the barycentric coordinates are $t_0(y)=(b-y)/(b-a)$ and $t_1(y)=(y-a)/(b-a)$, so the Bernstein representation is a polynomial in the single variable $y$.
For the triangle checks, we substitute an endpoint for $t=y_{uvw}$ and clear denominators to obtain a polynomial in the pair values $y_{uv},y_{uw},y_{vw}$.
We use $k=3$ on tetrahedra and smaller $k$ on degenerate intersections, working in affine coordinates on their affine hulls.

\begin{proof}[Computer-assisted proof of Lemma~\ref{lem:verification}]
We first describe how the verifier covers all pair and triangle configurations.  The rounding rule has the five closed one-dimensional pieces
{\small
\[
\begin{aligned}
  P_0&=\left[0,\frac12\right],& P_1&=\left[\frac12,\frac{25}{26}\right],& P_2&=\left[\frac{25}{26},1\right] &&\text{for positive pairs},\\
  N_0&=\left[0,\frac12\right],& N_1&=\left[\frac12,1\right] &&&&\text{for negative pairs}.
\end{aligned}
\]
}

These intervals cover every possible pair value $y\in[0,1]$, and the formulas on adjacent pieces agree at their common endpoint.  On each of the five pieces the verifier checks the pair inequality $D(e)-V(e)\ge0$, after clearing its positive denominator.  These are the five one-dimensional checks needed for the certificate.  As a redundant sanity check, the implementation also verifies the four elementary probability inequalities $p_0\ge0$, $1-p_0\ge0$, $p_1\ge0$, and $1-p_1\ge0$ on each piece.

We next consider triangles.  Up to relabeling their vertices, the four possible sign types are $+++$, $++-$, $+--$, and $---$.  Assigning each positive edge to one of $P_0,P_1,P_2$ and each negative edge to one of $N_0,N_1$, again up to relabeling edges of the same sign, gives
{\small
\[
  \binom{3+3-1}{3}=10,\qquad
  \binom{3+2-1}{2}\cdot2=12,\qquad
  3\binom{2+2-1}{2}=9,\qquad
  \binom{2+3-1}{3}=4
\]
}closed state boxes of the four respective sign types, for a total of $35$.  Consequently, every signed triangle and every choice of its three pair values in $[0,1]$ belongs to at least one of these boxes; the harmless overlaps at the rounding breakpoints ensure that no boundary value is omitted.

Consider any locally feasible triple value $t=y_{uvw}$ and write $y_1,y_2,y_3$ for the pair values $y_{uv},y_{uw},y_{vw}$.  By~\eqref{eq:joint-unified} and~\eqref{eq:triangle-variance}, $D(T)-V(T)$ is affine in $t$, so it is enough to prove nonnegativity at both endpoints of the interval~\eqref{eq:local-feas}.
The lower endpoint is the maximum of four affine functions,
$\max\{0,y_1+y_2-1,y_1+y_3-1,y_2+y_3-1\}$, and the upper endpoint is
the minimum of three, $\min\{y_1,y_2,y_3\}$.
To describe the two tetrahedral covers explicitly, let $e_1,e_2,e_3$ be the coordinate unit vectors, put $\mathbf1=(1,1,1)$, and let $m=\mathbf1/2$.
The region where the lower endpoint is zero splits along $y_1+y_2+y_3=1$ into the two tetrahedra $\operatorname{conv}\{0,e_1,e_2,e_3\}$ and $\operatorname{conv}\{m,e_1,e_2,e_3\}$.
The other three tetrahedra are $\operatorname{conv}\{e_i,e_j,m,\mathbf1\}$ for $1\le i<j\le3$, on which the lower endpoint equals $y_i+y_j-1$.
For the upper endpoint, the three tetrahedra are $\operatorname{conv}\{0,e_j,e_k,\mathbf1\}$, where $\{i,j,k\}=\{1,2,3\}$, on which it equals $y_i$.
This gives two covers of the feasible pair-value domain, that is, of the set of pair values for which the interval~\eqref{eq:local-feas} is nonempty: five tetrahedra for the lower endpoint and three for the upper endpoint, in each case with pairwise disjoint interiors.

For each of the $35$ boxes, the verifier intersects it with each of the eight endpoint tetrahedra; empty intersections are skipped.
Each intersection is defined by the six box inequalities and the four tetrahedron inequalities.
Its vertices are found by solving every nonsingular triple of boundary equations and retaining the solutions satisfying all ten inequalities, using exact rational arithmetic.
For an intersection of dimension three, the verifier
triangulates each polygonal face by drawing diagonals from one chosen vertex to all nonadjacent vertices,
then joins each resulting triangle to the average of all vertices of the intersection, which is an interior point.
Intersections of smaller dimension are triangulated in their affine hulls.

On each intersection, the verifier substitutes the corresponding endpoint for $t$.  Multiplying
$D(T)-V(T)$ by the positive quantity $\prod_e(2-y_e)$ clears all denominators and leaves a polynomial of total degree at most~$11$.
The polynomial is expressed on each resulting simplex in the Bernstein basis~\cite{bernstein_basis}.  Since the Bernstein basis functions are nonnegative and sum to one, nonnegativity of every Bernstein coefficient certifies nonnegativity throughout the simplex (Fact~\ref{fact:bernstein-enclosure}).  Whenever this test is inconclusive, the verifier bisects the simplex at the midpoint of a longest edge and repeats the test on both halves; reaching the depth or node limit would end the run with an inconclusive verdict, never with a false pass.

All checks pass with exact rational arithmetic.  Thus the pair inequalities hold for every signed pair value, and the triangle inequalities hold for all four sign types, all pair values, and every locally feasible triple value.  Only three of the $35$ triangle state boxes require any subdivision, to maximum depth $11$.
The verifier is available as \path{verify_rounding.py} at \url{https://github.com/elhipercubo/cluster-lp-rounding}; it is a single self-contained script that takes no arguments and uses only the Python standard library.
\end{proof}

%% file: hardness.tex
\section{Hardness of exact separation and exact optimization}\label{app:hardness}

We prove both hardness results from \textsc{Vertex Cover}.
We first allow an \emph{incomplete} signed instance $H=(W,F^+,F^-)$, in which pairs outside $F^+\cup F^-$ have weight zero.  Such pairs are ignored in the definitions of $\cost_H$ and of the cluster LP.

Let $J=(U,E)$ be an arbitrary graph, write $n=|U|$ and $m=|E|$, let $\tau(J)$ denote its minimum vertex-cover size, and add one new vertex $r$.  Define an incomplete unit-weight
signed instance $H(J)$ on $U\mathbin{\dot\cup}\{r\}$ by
\begin{equation}
  F^+(H(J))=\{ru:u\in U\},
  \qquad
  F^-(H(J))=E.
\end{equation}
All other pairs are absent.  Thus the cluster containing $r$ is rewarded for including many vertices, but pays for every edge of $J$ whose two endpoints it includes.

Consider first an ordinary clustering.  Splitting every cluster not containing $r$ into singletons does not cut any additional positive pair and can only remove negative pairs from
within clusters.  We may therefore assume that the clustering consists of $\{r\}\cup I$ and singleton clusters for the vertices in $U\setminus I$.  Its number of disagreements is
 $n-|I|+e_J(I)$.
 Notice that
\begin{equation}\label{eq:cover-cleaning}
  \min_{I\subseteq U}\bigl(n-|I|+e_J(I)\bigr)=\tau(J).
\end{equation}
Indeed, if $C$ is a vertex cover and $I=U\setminus C$, then $I$ is independent and the expression equals $|C|$.  Conversely, for arbitrary $I$, start with $U\setminus I$ and add one endpoint of every edge of $J[I]$.  The resulting set is a vertex cover of size at most $n-|I|+e_J(I)$, proving~\eqref{eq:cover-cleaning}.

\begin{lemma}\label{lem:hardness-star}
The cluster LP of $H(J)$ satisfies
\begin{equation}\label{eq:star-lp-value}
  \OPT_{\mathrm{LP}}(H(J))=\tau(J).
\end{equation}
Moreover, for an integer $k$, define $q^{(k)}\in\mathbb Q^{U\cup\{r\}}$ by
\begin{equation}
  q^{(k)}_u:=\frac12\quad(u\in U),
  \qquad
  q^{(k)}_r:=k-\frac{n-1}{2}.
\end{equation}
Then
\begin{equation}\label{eq:star-reduced-cost}
  \max_{S\subseteq U\cup\{r\}}
  \bigl(q^{(k)}(S)-\cost_{H(J)}(S)\bigr)
  =\max\left\{0,k-\tau(J)+\frac12\right\}.
\end{equation}
\end{lemma}

\begin{proof}
We have shown $\OPT_{\mathrm{LP}}(H(J))\le\tau(J)$. For the reverse inequality, define a
dual vector $\bar q$ by
\[
  \bar q_u:=\frac12\quad(u\in U),
  \qquad
  \bar q_r:=\tau(J)-\frac n2.
\]
Every subset of $U\cup\{r\}$ is either $I$ or $\{r\}\cup I$ for some $I\subseteq U$.  In the first case,
\begin{align*}
  \cost_{H(J)}(I)=e_J(I)+\frac{|I|}{2},\qquad \bar q(I)-\cost_{H(J)}(I)=-e_J(I)\le0.
\end{align*}
In the second case, by~\eqref{eq:cover-cleaning},
\begin{align*}
  \cost_{H(J)}(\{r\}\cup I)&=e_J(I)+\frac{n-|I|}{2},\\
  \bar q(\{r\}\cup I)-\cost_{H(J)}(\{r\}\cup I)
  &=\tau(J)-\bigl(n-|I|+e_J(I)\bigr)\le0.
\end{align*}
Hence $\bar q$ is dual feasible with objective value $\sum_{v\in U\cup\{r\}}\bar q_v=\tau(J)$, and weak duality proves~\eqref{eq:star-lp-value}.

For the second assertion,
    for a subset $I \subseteq U$ we have
    $
  q^{(k)}(I)-\cost_{H(J)}(I)=-e_J(I)\le0
  $
in the same manner,
whereas for a subset of the form $\{r\}\cup I$ we have
\begin{equation}\label{eq:root-cluster-reduced-cost}
  q^{(k)}(\{r\}\cup I)-\cost_{H(J)}(\{r\}\cup I)
  =k+\frac12-\bigl(n-|I|+e_J(I)\bigr).
\end{equation}
Taking the maximum in~\eqref{eq:root-cluster-reduced-cost} and using~\eqref{eq:cover-cleaning}
proves~\eqref{eq:star-reduced-cost}.
\end{proof}

\subsection{Balanced completion}
We next show how to remove the absent pairs while preserving both conclusions of the lemma.  At a high level, we replace each vertex by a positive clique of $L$ copies.  A
specified pair becomes a uniformly signed complete bipartite graph between the corresponding blocks.  An absent pair is replaced by a balanced bipartite pattern with as many
positive as negative pairs in every row and column.  If blocks are selected all-or-nothing, this pattern contributes $L^2/4$ to the cluster cost for each selected endpoint block.
The shift $a_u/4$ in the lifted dual weights cancels exactly these linear contributions, so absent pairs disappear from the reduced-cost function.

The main issue is that a candidate set could select a block partially.   The positive clique inside each block imposes a
quadratic penalty for doing so.  We choose the balanced patterns to have small spectral norm and take $L$  large enough that this penalty dominates every possible
discrepancy gain.  An optimal reduced-cost set can therefore be taken to be a union of entire blocks.

\paragraph{Construction.}
Let $H=(W,F^+,F^-)$ be an incomplete unit-weight signed instance, let
\[
  Z:=\binom W2\setminus(F^+\cup F^-),
\]
and let $a_u$ be the number of pairs of $Z$ incident to $u$.

Put $N=|W|$. Let $R$ be the smallest power of two satisfying $R\ge\max\{2,(N-1)^2\}$, set $L=2R\ge 2(N-1)^2$, and let $K\in\{-1,1\}^{R\times R}$ be a Hadamard matrix, so $KK^{\mathsf T}=RI$.  Define
\[
  M:=
  \begin{pmatrix}
    K&-K\\
    -K&K
  \end{pmatrix}
  = \begin{pmatrix}1&-1\\-1&1\end{pmatrix}\otimes K.
\]
Every row and column of $M$ sums to zero.
  An order-$R$ Hadamard matrix has spectral norm $\sqrt R$
  \cite[Section~14.3]{Jukna11}; since spectral norms multiply under Kronecker products,
\begin{equation}\label{eq:balanced-matrix-norm}
  \lVert M\rVert_2=2\sqrt R=\sqrt{2L}.
\end{equation}
Both $K$ and $M$ can be generated recursively in time polynomial in $N$.

Replace every $u\in W$ by a block $B_u$ of $L$ copies.  Make each block a positive clique.  Between $B_u$ and $B_v$, use all positive pairs if $uv\in F^+$ and all negative pairs if
$uv\in F^-$.  If $uv\in Z$, fix an ordering of its endpoints and declare the pair between the $i$th and $j$th copies positive when $M_{ij}=1$ and negative when $M_{ij}=-1$.  This defines the complete signed graph $\widehat H$.

\begin{lemma}\label{lem:balanced-completion}
The method above constructs, in polynomial time,
a complete unit-weight signed instance $\widehat H$ and an integer $L=O(|W|^2)$ with the following properties.

For every $q\in\mathbb R^W$, assign every copy $u_i$ of $u$ the value
\begin{equation}\label{eq:lifted-dual-vector}
  \widehat q_{u_i}:=L\left(q_u+\frac{a_u}{4}\right).
\end{equation}
Then
\begin{equation}\label{eq:completion-reduced-cost}
  \max_{X\subseteq V(\widehat H)}
  \bigl(\widehat q(X)-\cost_{\widehat H}(X)\bigr)
  =L^2\max_{T\subseteq W}\bigl(q(T)-\cost_H(T)\bigr).
\end{equation}
Moreover,
\begin{equation}\label{eq:completion-lp-value}
  \OPT_{\mathrm{LP}}(\widehat H)
  =L^2\left(\OPT_{\mathrm{LP}}(H)+\frac{|Z|}{2}\right).
\end{equation}
\end{lemma}

\begin{proof}
We first prove~\eqref{eq:completion-reduced-cost}.  For $X\subseteq V(\widehat H)$ write $X_u=X\cap B_u$, $x_u=|X_u|/L$, and
\[
  \Psi_q(x):=\sum_{u\in W}q_ux_u
  -\sum_{uv\in F^-}x_ux_v
  -\frac12\sum_{uv\in F^+}(x_u+x_v-2x_ux_v).
\]
This is the multilinear extension of $q(T)-\cost_H(T)$.  In particular,
\begin{equation}\label{eq:multiaffine-maximum}
  \max_{x\in[0,1]^W}\Psi_q(x)=\max_{T\subseteq W}\bigl(q(T)-\cost_H(T)\bigr),
\end{equation}
because a multi-affine function attains its maximum over a cube at a vertex.

For an absent pair $uv\in Z$, put
\[
  h_{uv}(X_u,X_v):=\sum_{i\in X_u}\sum_{j\in X_v}M_{ij}.
\]
The row and column balance of $M$ shows by direct counting that the contribution of the pairs between $B_u$ and $B_v$ to $\cost_{\widehat H}(X)$ is
\begin{equation}\label{eq:absent-pair-cost}
  \frac{L^2}{4}(x_u+x_v)-h_{uv}(X_u,X_v).
\end{equation}
The positive clique inside $B_u$ contributes $\frac{L^2}{2}x_u(1-x_u)$.  The uniform block pairs corresponding to $F^+$ and $F^-$ contribute $L^2$ times the corresponding terms in the multilinear extension of $\cost_H$.  Consequently, the shifts $a_u/4$ in~\eqref{eq:lifted-dual-vector} cancel the first term of~\eqref{eq:absent-pair-cost}, and
\begin{equation}\label{eq:completion-master-identity}
  \widehat q(X)-\cost_{\widehat H}(X)
  =L^2\Psi_q(x)
   +\sum_{uv\in Z}h_{uv}(X_u,X_v)
   -\frac{L^2}{2}\sum_{u\in W}x_u(1-x_u).
\end{equation}

Put $c_u=x_u(1-x_u)$.  Since $M\one=0$ and $\one^{\mathsf T}M=0$,
    $
  h_{uv}(X_u,X_v)
  =(\one_{X_u}-x_u\one)^{\mathsf T}M(\one_{X_v}-x_v\one).
  $
Moreover, $\lVert\one_{X_u}-x_u\one\rVert_2^2=Lc_u$.  The spectral bound~\eqref{eq:balanced-matrix-norm}, followed by the AM-GM inequality, gives
\begin{align}
  \sum_{uv\in Z}h_{uv}(X_u,X_v)
  &\le\sqrt2L^{3/2}\sum_{uv\in Z}\sqrt{c_uc_v}\notag\\
  &\le\frac{\sqrt2L^{3/2}}2\sum_{u\in W}a_uc_u
  \le\frac{\sqrt2L^{3/2}(N-1)}2\sum_{u\in W}c_u
  \le\frac{L^2}{2}\sum_{u\in W}c_u.
\end{align}
Together with~\eqref{eq:completion-master-identity} and~\eqref{eq:multiaffine-maximum}, this proves the upper bound in~\eqref{eq:completion-reduced-cost}.  Conversely, if $X=\bigcup_{u\in T}B_u$ is a union of whole blocks, then every $c_u$ and every $h_{uv}(X_u,X_v)$ is zero, so equality is attained for a maximizing $T$.  This proves~\eqref{eq:completion-reduced-cost}.

We finish by proving~\eqref{eq:completion-lp-value}.  If $q$ is feasible for the cluster-LP dual of $H$, then~\eqref{eq:completion-reduced-cost} shows that the lifted vector~\eqref{eq:lifted-dual-vector} is feasible for the dual of $\widehat H$.  Its objective value is
\[
  \sum_{u\in W}\sum_{i\in B_u}\widehat q_{u_i}
  =L^2\left(\sum_{u\in W}q_u+\frac14\sum_{u\in W}a_u\right)
  =L^2\left(\sum_{u\in W}q_u+\frac{|Z|}{2}\right).
\]
Taking an optimal $q$ proves the lower bound in~\eqref{eq:completion-lp-value}.

For the reverse bound, let $\widehat q$ be any feasible dual solution for $\widehat H$, put
\[
  Q_u:=\sum_{i\in B_u}\widehat q_{u_i},
  \qquad
  q_u:=\frac{Q_u}{L^2}-\frac{a_u}{4},
\]
and, for $T\subseteq W$, write $B_T=\bigcup_{u\in T}B_u$.  By construction,
\begin{equation}\label{eq:full-block-cost}
  \cost_{\widehat H}(B_T)
  =L^2\left(\cost_H(T)+\frac14\sum_{u\in T}a_u\right).
\end{equation}
Indeed, every specified pair contributes $L^2$ times its old cost, while an absent pair contributes $L^2/2$ if both blocks are selected and $L^2/4$ if exactly one is selected.  Applying feasibility of $\widehat q$ to $B_T$ and using~\eqref{eq:full-block-cost} gives $q(T)\le\cost_H(T)$.  Thus $q$ is dual feasible for $H$, and
\[
  \sum_{u,i}\widehat q_{u_i}
  =L^2\left(\sum_uq_u+\frac{|Z|}{2}\right)
  \le L^2\left(\OPT_{\mathrm{LP}}(H)+\frac{|Z|}{2}\right).
\]
Maximizing over $\widehat q$ proves the reverse bound and completes the proof.
\end{proof}

\subsection{Hardness consequences}

\begin{proof}[Proof of Theorem~\ref{thm:hardness-intro}]
We first consider exact dual separation.  Given an instance $(J,k)$ of \textsc{Vertex Cover}, construct $H(J)$ and $q^{(k)}$ in polynomial time as in Lemma~\ref{lem:hardness-star}, and then apply Lemma~\ref{lem:balanced-completion} to obtain $\widehat H(J)$ and the lifted vector $\widehat q^{(k)}$.  Equations~\eqref{eq:star-reduced-cost} and~\eqref{eq:completion-reduced-cost} imply
\[
  \exists X\subseteq V(\widehat H(J)):\
  \widehat q^{(k)}(X)>\cost_{\widehat H(J)}(X)
  \quad\Longleftrightarrow\quad
  \tau(J)\le k,
\]
proving NP-hardness of the violated-constraint decision problem under polynomial-time many-one reductions.  It belongs to NP because a violating set gives a polynomially-checkable certificate.  Exact dual feasibility is its complementary decision problem and is therefore coNP-complete.

Now we show NP-completeness of the fixed cluster-LP objective. The problem belongs to NP:  if a cluster-LP solution of value at most a given rational threshold exists, then there
is such a basic feasible solution supported on at most $|V(\widehat H)|$ clusters.  These can be listed explicitly, and their rational coefficients have polynomial encoding length
by Cramer's rule.
This sparse solution is a polynomial-size certificate whose feasibility and value can be checked exactly.

As for hardness,
write $n=|U|$ and $m=|E|$.  The absent pairs of $H(J)$ are precisely the nonedges of $J$ inside $U$, so
$
  |Z|=\binom n2-m.
  $
Equations~\eqref{eq:star-lp-value} and~\eqref{eq:completion-lp-value} give
\[
  \OPT_{\mathrm{LP}}(\widehat H(J))
  =L^2\left(\tau(J)+\frac12\left(\binom n2-m\right)\right).
\]
Therefore,
\[
  \tau(J)\le k
  \quad\Longleftrightarrow\quad
  \OPT_{\mathrm{LP}}(\widehat H(J))
  \le L^2\left(k+\frac12\left(\binom n2-m\right)\right).
\]
This proves NP-hardness of solving the cluster LP on complete unweighted instances.

\end{proof}

%% file: query_lb.tex
\subsection{A query lower bound for approximate separation}\label{sec:query_lb}
We observe that the quadratic dependence on the number of vertices is unavoidable for a general approximate separator in the adjacency-list query model, even though it can be avoided when one only seeks
an approximately optimal integral primal solution:

\begin{proposition}\label{prop:separation-query-lower-bound}
Fix $\varepsilon>0$.  Any randomized algorithm that, given an explicit vector $q\in\mathbb Q^V$ and adjacency-list access (degree, neighbor and pair queries) to the positive graph of a complete unweighted signed
instance, with probability at least $2/3$ either finds $S\subseteq V$ with $q(S)>\cost(S)$ or correctly reports that
\[
  \frac{q(S)}{1+\varepsilon}\le\cost(S)\qquad\text{for every }S\subseteq V
\]
requires $\Omega(n^2)$ queries in the worst case. The lower bound holds even when $q\ge0$.
The same lower bound therefore holds for the runtime of RAM algorithms with explicit adjacency-list input.
\end{proposition}

\begin{proof}
Let $n=2k$, partition $V=K\mathbin{\dot\cup}R$ with $|K|=|R|=k$, fix $u_0\in K$, and set $q_{u_0}=1/2$ and $q_u=0$ for $u\ne u_0$.  We may even reveal the partition
$K\mathbin{\dot\cup}R$ to the algorithm.  In the first instance $G_0$, the positive graph is the disjoint union of the two cliques on $K$ and $R$; then $\cost_{G_0}(K)=0<q(K)$, so the separator must report a violation.
For the second instance $G_1$, choose distinct $a,b\in K\setminus\{u_0\}$ and $c,d\in R$, and switch the signs of four pairs:
\[
  ab,cd:+\longrightarrow- ,
  \qquad
  ac,bd:-\longrightarrow+ .
\]
Every positive degree remains $k-1$, while the positive graph becomes connected.  Hence every nonempty proper $S\subset V$ has a positive pair crossing its boundary, and
$\cost_{G_1}(S)\ge1/2$.  The set $V$ contains negative pairs and also has cost at least $1/2$.  Since $q(S)\le1/2$ for every $S$, the vector $q$ itself is dual feasible for $G_1$,
    so the separator must report approximate feasibility.

Fix any ordering of the positive adjacency lists of $G_0$, choose the
four switched pairs uniformly subject to the conditions above, and
obtain the lists of $G_1$ by replacing, in the same positions, the
deleted positive neighbors by the added ones.
All degrees are identical, and all but four neighbor queries, as well as all but four pair queries, give the same answer, hidden among $\Theta(n^2)$ possible queries.  Until one of these entries or pairs is queried, the
two transcripts are identical.  Thus $o(n^2)$ queries give a distinguishing advantage only of $o(1)$.  Yao's minimax principle gives the claimed $\Omega(n^2)$ randomized lower bound.
\end{proof}

%% file: seeded_cc.tex
\subsection{Seeded correlation clustering}\label{app:seeded-cc}
As defined in the introduction, the input consists of a signed graph and a set $R\subseteq V$ of seeds, and every cluster must contain at most one seed.  Equivalently, this is the special case of constrained correlation clustering~\cite{vanZuylen+09} with hostile set $\binom R2$ and no friendly pairs.

Let
$
  \mathcal F_R=\{S\subseteq V:|S\cap R|\le1\}.
  $
The \emph{seeded cluster LP} and its dual are the specialization of the constrained cluster LP of~\cite{towards2} for $F = \emptyset$, $H = \binom{R}{2}$:
\begin{align}
  \min\quad &\sum_{S\in\mathcal F_R}\cost(S)z_S
  &\text{s.t.}\quad&
  \sum_{\substack{S\in\mathcal F_R\\u\in S}}z_S=1\quad(u\in V),
  &z_S&\ge0\quad(S\in\mathcal F_R),                             \label{eq:seeded-cluster-primal}\\
  \max\quad &\sum_{u\in V}q_u
  &\text{s.t.}\quad&
  q(S)\le\cost(S)\quad(S\in\mathcal F_R),
  &q_u&\in\mathbb R\quad(u\in V).                               \label{eq:seeded-cluster-dual}
\end{align}

\begin{theorem}[Seeded cluster LP]\label{thm:seeded-cluster-lp}
For every \(R\subseteq V\) and every \(\varepsilon>0\), there is a deterministic \(2^{\poly(1/\varepsilon)}n^{O(1)}\)-time algorithm
computing $(1+\varepsilon)$-approximate solutions $q, z$ to~\eqref{eq:seeded-cluster-dual} and~\eqref{eq:seeded-cluster-primal}.
\end{theorem}

\begin{proof}
We indicate the few changes to the proof of Theorem~\ref{thm:main}.  For dual separation, let \(S^*\) minimize \(\cost(S)/q(S)\) over \(S\in\mathcal F_R\) with \(q(S)>0\).
The family \(\mathcal F_R\) is downwards closed:
\[
  S\in\mathcal F_R,\quad A\subseteq S
  \quad\Longrightarrow\quad A\in\mathcal F_R .
\]
Therefore both sides of every partition \(S^*=A\mathbin{\dot\cup}B\) are admissible, so the proof of Lemma~\ref{lem:structure}\ref{it:half-density} still gives
\[
  e^+(A,B)\ge\frac{|A||B|}{2}.
\]
The remaining identities used by localization hold for arbitrary sets, while the positive-mass argument uses only singletons, which are admissible.  Hence the large-cost argument and
the threshold-localization construction are unchanged; the localized universe itself may contain several seeds.

In the small-cost regime, modify the fixed-cardinality optimization following~\eqref{eq:fixed-cardinality} by maximizing the retained score over admissible sets.  Namely, take the best
option between (a) the
\(s\) highest-scoring vertices of \(U\setminus R\), and (b) the best-scoring vertex of \(U\cap R\) together with the \(s-1\) highest-scoring vertices of \(U\setminus R\).  Since \(S^*\subseteq U\) is feasible, the proof of Lemma~\ref{lem:near} is unchanged.

For intermediate-cost recovery, the preliminary test for a negative coefficient \(b_u\) is unchanged because every singleton is admissible.  We now solve  the constrained version of Lemma~\ref{lem:dense-quadratic} with \(b\ge0\):
\[
  \min\left\{\frac12x^{\mathsf T}Ax+b^{\mathsf T}x:
  x\in\{0,1\}^U,\ \sum_{r\in R\cap U}x_r\le1\right\}.
\]
Refine the signature of every vertex \(u\) by the bit \(\one[u\in R]\), so that every resulting aggregate class \(C\in\mathcal C\) is contained in \(R\) or disjoint from it, and add to~\eqref{eq:aggregated-statistic-LP} the single constraint
\[
  \sum_{\substack{C\in\mathcal C\\C\subseteq R}}y_C\le1.
\]
This only doubles the number of classes.  An optimal basic solution now has at most \(2D+1\) variables strictly between their bounds.  Increasing the small-\(n_0\)
    exhaustive-search cutoff so that \(\Delta\ge2D+1\), and again rounding every class count down, preserves the seed constraint, cannot increase the linear objective, and changes every cut count by at most \(\Delta\).  The approximation proof and both running-time bounds of Lemma~\ref{lem:dense-quadratic} therefore remain unchanged asymptotically.

The zero-cost test scans the positive-edge components $K$ that are cliques and contain at most one seed, and returns any such $K$ with $q(K) > 0$.
With these modifications, every candidate returned
is admissible and the proof of Theorem~\ref{thm:main} applies verbatim.

Finally, as in Corollary~\ref{cor:lp-solutions}, we apply the ellipsoid algorithm to the dual LP with constraints in \(\mathcal F_R\).  Singleton columns still give feasibility and the same polynomial
bounds.  The zero optimum is recognizable: if a positive-edge component \(C\) is not a clique, put \(q_u=1/(2n)\) on \(C\) and zero elsewhere; if all components are cliques but some
\(C\) contains two seeds, put \(q_r=1/2\) on one such seed and zero elsewhere.  In either case every admissible set of positive \(q\)-weight has cost at least \(q(S)\).  Thus the same ellipsoid argument yields the stated primal and dual solutions.
\end{proof}

\begin{corollary}\label{cor:seeded-cc}
For every \(\varepsilon>0\), seeded correlation clustering admits a randomized \((1.92+\varepsilon)\)-approximation that succeeds with high probability and runs in
\(2^{\poly(1/\varepsilon)}n^{O(1)}\) time.\end{corollary}

\begin{proof}
Take the hostile set to be \(\binom R2\) and the friendly set to be empty.  Theorem~\ref{thm:seeded-cluster-lp} supplies
the polynomial-support nearly-optimal constrained
cluster-LP solution assumed by Theorem~4.18 of~\citet*{towards2} for constrained correlation clustering; every cluster in its support is admissible.  Applying their (randomized) rounding theorem
with suitably small accuracy
gives the claim.
\end{proof}

%% file: cluster_insertion.tex
\subsection{Cluster insertion}\label{app:cluster-insertion}

Let \(G=(V,E^+\mathbin{\dot\cup}E^-)\) be a complete signed instance, let \(w_{uv}>0\) be pair weights, and let \(\mathcal P\) be the current clustering.  Its weighted disagreement cost is
\[
  \cost_w(\mathcal P)=
  \sum_{\substack{uv\in E^+\\u\not\sim_{\mathcal P}v}}w_{uv}
  +\sum_{\substack{uv\in E^-\\u\sim_{\mathcal P}v}}w_{uv}.
\]
Write \(\cost_w^+(\mathcal P)\) and \(\cost_w^-(\mathcal P)\) for the two respective contributions.  For \(\emptyset\neq S\subseteq V\), write
\[
  \mathcal P+S=\{S\}\cup\{C\setminus S:C\in\mathcal P,\ C\setminus S\ne\varnothing\}
\]
for the clustering obtained by inserting \(S\).  Write \(u\sim_{\mathcal P}v\) when \(u\) and \(v\) currently belong to the same cluster and define the auxiliary pair weights
\begin{equation}\label{eq:insertion-auxiliary-weights}
  \widehat w_{uv}=w_{uv}\bigl(1+\one[u\sim_{\mathcal P}v]\bigr),
\end{equation}
preserving the sign of every pair. The associated cluster-column cost is
\begin{equation}\label{eq:insertion-auxiliary-cost}
  \widehat\cost(S)=
  \sum_{\substack{uv\in E^-\\u,v\in S}}\widehat w_{uv}
  +\frac12\sum_{\substack{uv\in E^+\\|\{u,v\}\cap S|=1}}\widehat w_{uv}.
\end{equation}
Finally, assign each vertex the nonnegative weight
\begin{equation}\label{eq:insertion-vertex-weight}
  q_u=
  \frac12\sum_{\substack{v:uv\in E^+\\u\not\sim_{\mathcal P}v}}w_{uv}
  +\sum_{\substack{v:uv\in E^-\\u\sim_{\mathcal P}v}}w_{uv}.
\end{equation}

\begin{lemma}[Cluster insertion is dual separation]\label{lem:insertion-identity}
For every \(S\subseteq V\),
\begin{equation}\label{eq:insertion-identity}
  \cost_w(\mathcal P+S)-\cost_w(\mathcal P)=\widehat\cost(S)-q(S).
\end{equation}
\end{lemma}
\begin{proof}
It suffices to check one pair \(uv\).  If \(uv\) is positive and currently internal, selecting exactly one endpoint creates cost \(w_{uv}\); its auxiliary weight is \(2w_{uv}\), so~\eqref{eq:insertion-auxiliary-cost} charges exactly \(w_{uv}\).  If it is positive and currently external, selecting both endpoints saves \(w_{uv}\), while selecting only one changes nothing; the auxiliary boundary charge is \(w_{uv}/2\) and each selected endpoint contributes \(w_{uv}/2\) to \(q\).  If \(uv\) is negative and currently internal, selecting exactly one endpoint saves \(w_{uv}\); its auxiliary weight is \(2w_{uv}\) and each endpoint contributes \(w_{uv}\) to \(q\).  Finally, if it is negative and currently external, only selecting both endpoints changes the cost, increasing it by \(w_{uv}\), exactly its auxiliary internal cost.  Summing over all pairs proves the identity.
\end{proof}

Thus an improving insertion is precisely a violated cluster-dual constraint for the auxiliary weighted instance.
In the iterated-flipping algorithm, negative pairs have weight \(1\) and positive pairs have weight in \([1,2]\).  Therefore each weight in~\eqref{eq:insertion-auxiliary-weights} lies in \([1,4]\).

We record here the extension of Theorem~\ref{thm:main} needed above.

\begin{theorem}[Bounded-weight approximate separation]\label{thm:bounded-weight-separation}
Fix \(B\ge1\).  Let every pair of a complete signed instance have a rational weight \(a_{uv}\in[1,B]\), and define
\[
  \cost_a(S)=
  \sum_{\substack{uv\in E^-\\u,v\in S}}a_{uv}
  +\frac12\sum_{\substack{uv\in E^+\\|\{u,v\}\cap S|=1}}a_{uv}.
\]
Given \(q\in\mathbb Q^V\) and \(\eta>0\), a deterministic algorithm either finds \(S\subseteq V\) with \(q(S)>\cost_a(S)\), or certifies that
\[
  \frac{q(S)}{1+\eta}\le\cost_a(S)\qquad\text{for every }S\subseteq V.
\]
Its running time is \(2^{\poly_B(1/\eta)}(n+\langle a\rangle+\langle q\rangle)^{O(1)}\).
\end{theorem}
\begin{proof}
We indicate the changes to the proof of Theorem~\ref{thm:main}.  We may assume $0<\eta\le1$ and set $\beta=\eta/10$.  If some set has positive $q$-weight, let $S^*$ minimize $\cost_a(S)/q(S)$ among such sets, and write $s=|S^*|$, $c=\cost_a(S^*)$, and $\rho^*=c/q(S^*)$; otherwise the conclusion is immediate.  For disjoint sets, let $a^\pm(A,B)$ denote the total $a$-weight of the positive or negative pairs across their cut.  The weighted partition identity is
\begin{equation}\label{eq:weighted-partition-identity}
  \cost_a(A)+\cost_a(B)
  =\cost_a(A\cup B)+a^+(A,B)-a^-(A,B).
\end{equation}
As in Lemma~\ref{lem:structure}\ref{it:half-density}, optimality of $S^*$ implies $a^+(A,B)\ge a^-(A,B)$ whenever $S^*=A\mathbin{\dot\cup}B$.  Since the pair weights lie in $[1,B]$,
\begin{equation}\label{eq:weighted-cut-density}
  e^+(A,B)\ge\frac{|A||B|}{B+1}.
\end{equation}
Thus $S^*$ is $\delta$-cut-dense in the positive-support graph for $\delta=1/(B+1)$.

The other structural estimates used in the main proof become
\begin{align}
  \sum_{v\in T}|\Gamma^+(v)\mathbin\triangle T|
  &\le2\cost_a(T),                                                   \label{eq:weighted-pivot-bound}\\
  \cost_a(X)&\le\cost_a(T)+B|X|\,|T\setminus X|
  &&(X\subseteq T),                                                  \label{eq:weighted-deletion-bound}\\
  \rho^*q^+(S^*)&\le B\binom{s}{2}+c.                                \label{eq:weighted-positive-mass}
\end{align}
The first holds because every mismatching pair has weight at least one, the second follows from~\eqref{eq:weighted-partition-identity}, and the third follows by applying optimality to the singletons of positive $q$-weight and summing their weighted positive degrees.

The only structural change in localization is the smaller cut density.  Put $\zeta=\vartheta=\delta^2/4$.  In the intermediate regime $c<B\binom{s}{2}/\beta$, define
\[
  P_s=\left\{u:d^+(u)\le s-1+\frac{Bs}{\beta\zeta}\right\},
  \qquad
  D_0=1+\frac{B}{\beta\zeta}.
\]
Equation~\eqref{eq:weighted-pivot-bound} gives $|S^*\setminus P_s|<\zeta s$, and every $u\in P_s$ satisfies $|\Gamma^+(u)|\le D_0s$.  Lemma~\ref{lem:threshold-localization}, applied with parameters $(\delta,\eta,\theta,D)=(\delta,\zeta,\vartheta,D_0)$, therefore produces a candidate universe containing $S^*$ of size
\[
  D_0\left(1+\frac{D_0}{\vartheta}\right)^{j+1}s
  =O_{B,\beta}(s),
\]
where $j$ depends only on $B$.

The three regime arguments now proceed as before with $B$-dependent constants.  In the large-cost regime,~\eqref{eq:weighted-positive-mass} gives the singleton argument with $B\binom{s}{2}$ in place of $\binom{s}{2}$.  In the small-cost regime $0<c\le\tau s^2$, where now $\tau<\eta/(64B(B+1)^2)$, define the neighborhood~\eqref{eq:near-U} with threshold $t/(2(B+1))$ in place of $t/3$.  By~\eqref{eq:weighted-pivot-bound} some $v\in S^*$ has $|\Gamma^+(v)\mathbin\triangle S^*|\le2c/s\le s/(8(B+1))$, and~\eqref{eq:weighted-cut-density} applied to singletons gives $d_{S^*}(u)\ge(s-1)/(B+1)$ for every $u\in S^*$; since $\tfrac12\le c\le\tau s^2$ forces $s\ge4$, the argument of Lemma~\ref{lem:near} yields $S^*\subseteq U$ and, as every crossing positive pair has weight at least one, $|U\setminus S^*|\le8(B+1)c/s$.  For the fixed-cardinality linearization, applying~\eqref{eq:weighted-partition-identity} to $U=(U\setminus D)\mathbin{\dot\cup}D$ gives
\[
  L_a(U\setminus D)=L_a(U)+\sum_{u\in D}\Bigl(d^+_{a,U}(u)-d^-_{a,U}(u)-\tfrac12d^+_a(u)+q_u\Bigr)+a^-(D)-a^+(D),
\]
where $L_a=\cost_a-q$, $d^\pm_{a,U}(u)$ is the $a$-weight of the positive or negative pairs between $u$ and $U\setminus\{u\}$, $d^+_a(u)=d^+_{a,V}(u)$, and $a^\pm(D)$ is the $a$-weight of the positive or negative pairs inside $D$.  Keeping the $t=s$ vertices of $U$ with the largest linear coefficients and discarding the last two terms loses at most $2B\binom{|U\setminus S^*|}{2}\le64B(B+1)^2c^2/s^2<\eta c$, as in Lemma~\ref{lem:near}.  In the intermediate regime, the preceding localization bound keeps the weak-regularity error proportional to $s^2$.  The localized Lagrangian has off-diagonal coefficients in $[-B,B]$; after division by $B$, Lemma~\ref{lem:dense-quadratic} applies.  Choosing the regime thresholds and internal accuracies with the corresponding $B$-dependent factors proves the guarantee and the stated running time.
\end{proof}

\begin{theorem}\label{thm:approximate-insertion}
Let \(w_{uv}=1\) on negative pairs and \(w_{uv}\in[1,2]\) on positive pairs.  Given a clustering \(\mathcal P\) and \(\eta>0\), a deterministic algorithm either returns a genuinely improving insertion \(S\), or certifies that, for every \(S\subseteq V\),
\begin{equation}\label{eq:pointwise-insertion-certificate}
  \cost_w(\mathcal P)-\cost_w(\mathcal P+S)
  \le\theta q(S),
  \qquad
  \theta=\frac{\eta}{1+\eta},
\end{equation}
where \(q\) is defined by~\eqref{eq:insertion-vertex-weight}.  It runs in \(2^{\poly(1/\eta)}n^{O(1)}\) time when the weights have polynomial length.
\end{theorem}

\begin{proof}
Apply Theorem~\ref{thm:bounded-weight-separation} with \(B=4\) to the auxiliary instance~\eqref{eq:insertion-auxiliary-weights}.  By~\eqref{eq:insertion-identity}, a violated
constraint is an improving insertion; otherwise \(\widehat\cost(S)\ge q(S)/(1+\eta)\) for every \(S\), giving~\eqref{eq:pointwise-insertion-certificate}.
\end{proof}

%% file: rounding_derandomization.tex
\section{Derandomizing the rounding algorithm}\label{app:rounding-derandomization}
For completeness, we record the standard conditional-expectation argument underlying the deterministic version of Theorem~\ref{thm:rounding}.

\begin{proposition}\label{prop:rounding-derandomization}
Given an explicitly represented feasible cluster-LP solution $z$, one can deterministically construct a clustering $\mathcal A$ satisfying
$
  \cost(\mathcal A)\le\alpha_0\sum_S\cost(S)z_S
  $
  in polynomial time.
\end{proposition}

\begin{proof}
We derandomize the two procedures separately, then take the cheaper of the two.  For cluster-based rounding, if $W$ is the alive set and $P$ is the cost already incurred, use the potential
\[
  \Phi_{\mathcal C}:=P+\sum_{e\in\binom W2}\kappa_e.
\]
Under randomized cluster-based rounding, the probability that the final clustering causes a disagreement with edge $e$
    is $\kappa_e$ by~\eqref{eq:clock-cost}.
Therefore $\Phi_{\mathcal C}$ is exactly the
conditional expected final cost, given the choices already made.
    Thus, there exists some set $S$ in the support with $S\cap W\ne\emptyset$ that does not increase it; choose such a set deterministically and output $S\cap W$.
     Repetition gives a clustering $\mathcal C$ with
     $
  \cost(\mathcal C)\le\sum_e\kappa_e.
  $

  Likewise, for pivot rounding use the potential
\[
  \Phi_{\mathcal P}:=P+\sum_{e\in\binom W2}B_e.
\]
The one-round inequality established in the proof of
Lemma~\ref{lem:variance-certificate} says exactly that the expected value
of $\Phi_{\mathcal P}$ after the next complete pivot round is at most its
current value.
Fix the random choices in that round successively: first the pivot $u$,
then the sampled LP cluster $S\ni u$, and finally the independent joining
coins.  At every step, select an outcome for which the conditional
expectation of the post-round potential is no larger than before that
choice.  The resulting cluster contains $u$; repeating this process produces a clustering $\mathcal P$ with
$
  \cost(\mathcal P)\le\sum_eB_e.
  $
Finally, by~\eqref{eq:budget},
$\min\{\cost(\mathcal C),\cost(\mathcal P)\}\le\omega\sum_e\kappa_e+(1-\omega)\sum_eB_e=\alpha_0\sum_e\ell_e=\alpha_0\sum_S\cost(S)z_S$.

All conditional expectations above are computable in polynomial time.  In both procedures $S$ ranges over the support of $z$, and the potential after choosing $S$ in cluster-based rounding is explicit.  In pivot rounding, after $u$ and $S$ are fixed, the potential is a sum of functions of pairs of joining indicators, and the unfixed indicators remain independent Bernoulli variables with the known probabilities $p_0,p_1$; hence the conditional expectation of each term, with the fixed indicators replaced by their values, is a product of at most two of these probabilities.  Expectations before choosing $u$ or $S$ are obtained by averaging over their polynomially many possibilities, with weights $1/|W|$ and $z_S$ respectively.
\end{proof}

%% file: covering_gap.tex
\section{The cluster LP differs from its covering relaxation}\label{app:covering-gap}

For comparison with the standard cluster LP~\eqref{eq:cluster-primal}, define its direct covering relaxation by replacing every constraint $\sum_{S\ni u}z_S=1$ with $\sum_{S\ni
    u}z_S\ge1$.  Write the optimum values of these two programs as $\mathrm{LP}_{=}$ and $\mathrm{LP}_{\ge}$, respectively.  Their duals have the same constraints $q(S)\le\cost(S)$
    for every $S\subseteq V$, but the former allows $q\in\mathbb R^V$, whereas the latter requires $q\ge0$.  The next example shows that this distinction is genuine:
there is a complete unweighted instance on eight vertices for which
    $
  \mathrm{LP}_{\ge}=\frac{35}{4}<9=\mathrm{LP}_{=}.
  $
Since a nonnegative optimum standard-dual solution of value $9$ would also be feasible for the covering dual,
      this implies that every optimum solution of the standard cluster-LP dual has a negative coordinate (in this example).

Let
\[
  V=A\mathbin{\dot\cup}B\mathbin{\dot\cup}D\mathbin{\dot\cup}\{h,r\},
  \qquad
  A=\{a_1,a_2\},\quad B=\{b_1,b_2\},\quad D=\{d_1,d_2\}.
\]
The positive pairs are
\[
  E^+=\{rv:v\ne r\}\ \cup\ \binom{A\cup\{h\}}2\ \cup\ \{uv:u\in A\cup B,\ v\in D\},
\]
and all remaining pairs are negative.  Thus $r$ is positive to every other vertex, $A\cup\{h\}$ is a positive triangle, and all pairs between $A\cup B$ and $D$ are positive.

For the primal covering LP, assign value $1/2$ to each of the five sets
\[
  \{b_1\},\quad \{b_2\},\quad A\cup\{h,d_1,r\},\quad A\cup\{h,d_2,r\},\quad B\cup D\cup\{r\}.
\]
This solution is feasible and has value $35/4$.  For the equality LP, the integral solution with two clusters $\{b_1\}$ and $V\setminus\{b_1\}$
has value $9$.

It remains to certify that these upper bounds are tight.  Consider the two dual vectors
\[
\begin{array}{c|ccccc}
 &a_i&b_i&d_i&h&r\\ \hline
 q^{\ge}&7/8&3/2&5/4&3/2&0\\
 q^{=}&1&3/2&3/2&3/2&-1/2.
\end{array}
\]
By direct computation, both vectors are feasible for their respective duals, with values $35/4$ and $9$.